\documentclass[submission]{eptcs}

\usepackage{iftex}

\ifpdf
  \usepackage{underscore}         % Only needed if you use pdflatex.
  \usepackage[T1]{fontenc}        % Recommended with pdflatex
\else
  \usepackage{breakurl}           % Not needed if you use pdflatex only.
\fi

\usepackage{mathtools}
\usepackage{amssymb}
\usepackage{amsfonts}

\allowdisplaybreaks

\usepackage{amsthm}
\newtheorem{lemma}{Lemma}

\newtheorem{corollary}{Corollary}
\newtheorem{definition}{Definition}
\newtheorem{remark}{Remark}

\usepackage{tikz}
\usetikzlibrary{arrows}
\usetikzlibrary{fit}
\usetikzlibrary{positioning}

\usepackage{tikzit}
\input{quantum.tikzdefs}
\tikzstyle{gate}=[shape=rectangle, text height=1.5ex, text depth=0.25ex, yshift=0.5mm, fill=white, draw=black, minimum height=5mm, yshift=-0.5mm, minimum width=5mm, font={\small}, tikzit category=circuit]
\tikzstyle{big gate}=[shape=rectangle, text height=1.5ex, text depth=0.25ex, yshift=0.5mm, fill=white, draw=black, minimum height=10mm, yshift=-0.5mm, minimum width=5mm, font={\small}, tikzit category=circuit]

\tikzstyle{Z dot}=[inner sep=0mm, minimum size=2mm, shape=circle, draw=black, fill={rgb,255: red,221; green,255; blue,221}, tikzit category=zx]
\tikzstyle{Z phase dot}=[minimum size=1.2em, font={\footnotesize\boldmath}, shape=rectangle, rounded corners=0.5em, inner sep=0.2em, outer sep=-0.2em, scale=0.8, draw=black, fill={rgb,255: red,221; green,255; blue,221}, tikzit shape=circle, tikzit draw=blue, tikzit category=zx]
\tikzstyle{Z box}=[Z phase dot, rounded corners=0, fill={rgb,255: red,221; green,255; blue,221}, tikzit shape=rectangle, tikzit draw=blue, tikzit category=zx]

\tikzstyle{X dot}=[Z dot, shape=circle, draw=black, fill={rgb,255: red,232; green,165; blue,165}, tikzit category=zx]
\tikzstyle{X phase dot}=[Z phase dot, tikzit shape=circle, tikzit draw=blue, fill={rgb,255: red,232; green,165; blue,165}, font={\footnotesize\boldmath}, tikzit category=zx]

\tikzstyle{red dot}=[Z dot, shape=circle, draw=black, fill={rgb,255: red,255; green,0; blue,0}, tikzit category=zx]
\tikzstyle{red phase dot}=[Z phase dot, text=white, draw=black, fill={rgb,255: red,255; green,0; blue,0}, tikzit draw=blue, tikzit category=zx]

\tikzstyle{hadamard}=[fill=yellow, draw=black, shape=rectangle, inner sep=0.6mm, minimum height=1.5mm, minimum width=1.5mm, tikzit category=zx]
\tikzstyle{dtriangle}=[fill=yellow,draw=black,shape=isosceles triangle,shape border rotate=-90,isosceles triangle stretches=true,inner sep=0.8pt,minimum width=0.25cm,minimum height=2mm]
\tikzstyle{vtriang}=[fill=yellow,draw=black,shape=isosceles triangle,shape border rotate=180,isosceles triangle stretches=true,inner sep=0.8pt,minimum width=0.25cm,minimum height=2mm]
\tikzstyle{bspider}=[fill=black,draw=black,scale=1,shape=isosceles triangle,shape border rotate=-90,isosceles triangle stretches=true,inner sep=1pt,minimum width=0.4cm,minimum height=3mm]
\tikzstyle{dbspider}=[fill=black,draw=black,scale=1,shape=isosceles triangle,shape border rotate=90,isosceles triangle stretches=true,inner sep=1pt,minimum width=0.4cm,minimum height=3mm]

\tikzstyle{starv}=[fill=yellow,draw=black,shape=star,star points=6,star point ratio=1.74,inner sep=0.8pt,minimum height=2.75mm]
\tikzstyle{starh}=[fill=yellow,draw=black,shape=star,star points=6,star point ratio=1.74,inner sep=0.8pt,minimum height=2.75mm, rotate=90]

\tikzstyle{paulibox}=[fill={rgb,255: red,221; green,221; blue,255}, draw=black, shape=rectangle, inner sep=0.6mm, minimum height=5mm, minimum width=5mm, font={\footnotesize}, text height=1.5ex, text depth=0.25ex, tikzit category=zx]
\tikzstyle{vertex}=[inner sep=0mm, minimum size=1mm, shape=circle, draw=black, fill=black, tikzit category=misc]
\tikzstyle{vertex set}=[inner sep=0mm, minimum size=1mm, shape=circle, draw=black, fill=white, font={\footnotesize\boldmath}, tikzit category=misc]
\tikzstyle{small black dot}=[fill=black, draw=black, shape=circle, inner sep=0pt, minimum width=1.2mm, tikzit category=circuit]
\tikzstyle{cnot ctrl}=[fill=black, draw=black, shape=circle, inner sep=0pt, minimum width=1.2mm, tikzit category=circuit]
\tikzstyle{cnot targ}=[fill=white, draw=white, shape=circle, tikzit category=circuit, label={center:$\oplus$}, inner sep=0pt, minimum width=2.1mm, tikzit fill={rgb,255: red,102; green,204; blue,255}, tikzit draw=black]
\tikzstyle{ket}=[fill=white, draw=black, shape=regular polygon, regular polygon sides=3, regular polygon rotate=-30, scale=0.7, inner sep=1pt, tikzit category=circuit, tikzit shape=rectangle, tikzit fill=green]
\tikzstyle{bra}=[fill=white, draw=black, shape=regular polygon, regular polygon sides=3, regular polygon rotate=30, scale=0.7, inner sep=1pt, tikzit category=circuit, tikzit shape=rectangle, tikzit fill=red]
\tikzstyle{scalar}=[shape=rectangle, text height=1.5ex, text depth=0.25ex, yshift=0.5mm, fill=white, draw=black, minimum height=5mm, yshift=-0.5mm, minimum width=5mm, font={\small}]
\tikzstyle{clabel}=[fill=white, draw=none, shape=rectangle, tikzit fill={rgb,255: red,56; green,255; blue,242}, font={\footnotesize}, inner sep=1pt, tikzit category=labels]
\tikzstyle{empty diagram}=[draw={gray!40!white}, dashed, shape=rectangle, minimum width=1cm, minimum height=1cm, tikzit category=misc]
\tikzstyle{amap}=[fill=white, draw=black, shape=NEbox, tikzit category=asymmetric, tikzit fill=yellow, tikzit shape=rectangle]
\tikzstyle{amap conj}=[fill=white, draw=black, shape=NWbox, tikzit category=asymmetric, tikzit fill=green, tikzit shape=rectangle]
\tikzstyle{amap adj}=[fill=white, draw=black, shape=SEbox, tikzit category=asymmetric, tikzit fill=red, tikzit shape=rectangle]
\tikzstyle{amap trans}=[fill=white, draw=black, shape=SWbox, tikzit category=asymmetric, tikzit fill=orange, tikzit shape=rectangle]
\tikzstyle{astate}=[fill=white, draw=black, shape=NEtriangle, tikzit category=asymmetric, tikzit shape=circle, tikzit fill=yellow]
\tikzstyle{astate conj}=[fill=white, draw=black, shape=NWtriangle, tikzit category=asymmetric, tikzit shape=circle, tikzit fill=green]
\tikzstyle{astate adj}=[fill=white, draw=black, shape=SEtriangle, tikzit category=asymmetric, tikzit shape=circle, tikzit fill=red]
\tikzstyle{astate trans}=[fill=white, draw=black, shape=SWtriangle, tikzit category=asymmetric, tikzit shape=circle, tikzit fill=orange]

\tikzstyle{hadamard edge}=[-, dashed, dash pattern=on 2pt off 0.5pt, thick, draw={rgb,255: red,68; green,136; blue,255}]
\tikzstyle{star edge}=[-, dashed, dash pattern=on 2pt off 0.5pt, thick, draw={rgb,255: red,255; green,136; blue,68}]
\tikzstyle{box edge}=[-, dashed, dash pattern=on 2pt off 0.5pt, thick, draw={rgb,255: red,203; green,192; blue,225}]
\tikzstyle{brace edge}=[-, tikzit draw=blue, decorate, decoration={brace,amplitude=1mm,raise=-1mm}]
\tikzstyle{diredge}=[->]
\tikzstyle{double edge}=[-, double, shorten <=-1mm, shorten >=-1mm, double distance=2pt]
\tikzstyle{gray edge}=[-, {gray!60!white}]
\tikzstyle{pointer edge}=[->, very thick, gray]
\tikzstyle{boldedge}=[-, line width=1.6pt, shorten <=-0.17mm, shorten >=-0.17mm]
\tikzstyle{bidir edge}=[<->, very thick, draw={rgb,255: red,191; green,191; blue,191}]

\usepackage{caption}
\usepackage{subcaption}

\title{Speedy Contraction of ZX Diagrams with Triangles\\via Stabiliser Decompositions}

\author{
	Mark Koch$\null^{1}$ \and
	Richie Yeung$\null^{2,3}$ \and
	Quanlong Wang$\null^{2}$ \and
	\institute{$\null^{1}$Quantinuum, Terrington House, 13-15 Hills Road, CB2 1NL Cambridge, United Kingdom}
	\institute{$\null^{2}$Quantinuum, 17 Beaumont Street, Oxford, OX1 2NA, United Kingdom}
	\institute{$\null^{3}$University of Oxford, Oxford, United Kingdom}
}

\newcommand{\titlerunning}{Speedy Contraction of ZX Diagrams with Triangles via Stabiliser Decompositions}
\newcommand{\authorrunning}{M. Koch, R. Yeung \& Q. Wang}

\hypersetup{
  bookmarksnumbered,
  pdftitle    = {\titlerunning},
  pdfauthor   = {\authorrunning},
  pdfsubject  = {EPTCS},               % Consider adding a more appropriate subject or description
}

\begin{document}
\maketitle

\begin{abstract}

Recent advances in classical simulation of Clifford+T circuits make use of the ZX calculus to iteratively decompose and simplify magic states into stabiliser terms. 
We improve on this method by studying stabiliser decompositions of ZX diagrams involving the triangle operation.
We show that this technique greatly speeds up the simulation of quantum circuits involving multi-controlled gates which can be naturally represented using triangles.
We implement our approach in the QuiZX library \cite{kissinger2022classical,kissinger2022simulating} and demonstrate a significant simulation speed-up (up to multiple orders of magnitude) for random circuits and a variation of previously used benchmarking circuits.
Furthermore, we use our software to contract diagrams representing the gradient variance of parametrised quantum circuits, which yields a tool for the automatic numerical detection of the barren plateau phenomenon in ans\"atze used for quantum machine learning.
Compared to traditional statistical approaches, our method yields exact values for gradient variances and only requires contracting a single diagram. The performance of this tool is competitive with tensor network approaches, as demonstrated with benchmarks against the \texttt{quimb} library \cite{gray2018quimb}.

\end{abstract}

\section{Introduction}

The Clifford+T fragment of quantum mechanics is widely used in quantum computing due to its simplicity and ability to approximate any unitary operation to arbitrary precision \cite{dawson2005solovay,kitaev2002classical}.
As a result of this approximate universality it is widely believed that classical simulation of Clifford+T circuits has an exponential cost.
A surprisingly effective technique to simulate circuits with relatively low T-counts is contraction by stabiliser decomposition \cite{bravyi2016improved}, where a Clifford+T state is written as a linear combination of stabilisers which can be efficiently simulated according to the Gottesman-Knill theorem \cite{gottesman9807006talk,aaronson2004improved}.
While naively, this decomposition yields $2^n$ terms for a circuit with $n$ T gates, there are more efficient strategies that only require $2^{\alpha n}$ terms for $\alpha < 1$.
The best known strategy has worst case $\alpha \approx 0.396$ and was recently found in \cite{kissinger2022classical} by representing Clifford+T circuits as ZX diagrams and contracting them, which corresponds to strong simulation.
By interleaving stabiliser decompositions with additional ZX diagram simplifications, they obtained state of the art results in Clifford+T simulation.

A common challenge when working within the Clifford+T gate set is the representation of multi-controlled gates (for example the Toffoli) which occur ubiquitously in quantum algorithms.
They need to be decomposed into Clifford and T gates with a lot of work being put into finding the cheapest possible representations \cite{gosset2013algorithm,giles2013exact,miller2011elementary}.
Similarly, expressing multi-controlled gates in the vanilla ZX calculus requires a number of T-spiders.
However, as pointed out in \cite{ng2018completeness}, there are more elegant representations when permitting the \textit{triangle node} as a generator.
The triangle itself can be decomposed into a Clifford+T diagram with four T-spiders as follows \cite{coecke2017picturing}:
\begin{equation} \label{tri-T}
	\begin{tikzpicture}
	\begin{pgfonlayer}{nodelayer}
		\node [style=vtriang, rotate=180] (0) at (-0.05, 0) {};
		\node [style=none] (1) at (-1, 0) {};
		\node [style=none] (2) at (1, 0) {};
	\end{pgfonlayer}
	\begin{pgfonlayer}{edgelayer}
		\draw (1.center) to (2.center);
	\end{pgfonlayer}
\end{tikzpicture}
 ~=~ \textstyle{\frac{1}{2}}~ \begin{tikzpicture}
	\begin{pgfonlayer}{nodelayer}
		\node [style=none] (1) at (-0.75, 0) {};
		\node [style=none] (2) at (4.25, 0) {};
		\node [style=X dot] (3) at (0, 0) {};
		\node [style=Z phase dot] (4) at (1.25, 0.625) {$-\frac{\pi}{4}$};
		\node [style=Z phase dot] (5) at (2.5, 1.5) {$-\frac{\pi}{4}$};
		\node [style=X dot] (6) at (2.5, 0.625) {};
		\node [style=Z phase dot] (7) at (1.25, -0.625) {$\frac{\pi}{4}$};
		\node [style=Z phase dot] (8) at (2.5, -1.5) {$\frac{\pi}{4}$};
		\node [style=X dot] (9) at (2.5, -0.625) {};
		\node [style=Z dot] (10) at (3.5, 0) {};
	\end{pgfonlayer}
	\begin{pgfonlayer}{edgelayer}
		\draw (3) to (1.center);
		\draw [bend left] (3) to (4);
		\draw [bend right] (3) to (7);
		\draw (4) to (6);
		\draw [bend left] (6) to (10);
		\draw (10) to (2.center);
		\draw [bend right] (9) to (10);
		\draw (9) to (8);
		\draw (7) to (9);
		\draw (6) to (5);
	\end{pgfonlayer}
\end{tikzpicture}

\end{equation}
Using the approach from \cite{kissinger2022classical} we can thus simulate a diagram with $n$ T-spiders and $m$ triangles by decomposing it into $2^{\alpha n + \beta m}$ stabiliser terms with $\beta = 4\alpha \approx 1.586$.
This is suboptimal since the triangle only has stabiliser rank 2, yielding a trivial scaling factor of $\beta = 1$.
However, we can do even better than this and will show $\beta \approx 0.774$ in this paper.

We build on the work in \cite{kissinger2022simulating,kissinger2022classical} and implement our triangle stabiliser decompositions together with custom simplification tactics as an extension to the QuiZX library \cite{kissinger2022simulating,kissinger2022classical}.
Using our implementation, we demonstrate a significant simulation speed-up for circuits involving multi-controlled gates compared to \cite{kissinger2022classical}.
Concretely, we benchmark on random Clifford+T+CCZ circuits and a variation of the hidden-shift circuits introduced in \cite{bravyi2016improved} and show a multiple orders of magnitude improvement over QuiZX.

Besides aiding in representing multi-controlled gates, the triangle node also features prominently in the algebraic ZX-calculus \cite{wang2019algebraic} and ZXW-calculus \cite{shaikh2022sum, koch2022quantum, poor2023completeness}.
Thus, our implementation can additionally be used to contract a wide class of diagrams from those calculi, extending the usefulness of our method beyond classical simulation.
To demonstrate this, we revisit previous work of ours in \cite{wang2022differentiating} where we showed how to represent the gradient variance of a parametrised quantum circuit as an algebraic ZX diagram.
This variance is important since it characterises the existence of the 
barren plateau phenomenon \cite{mcclean2018barren} in quantum machine learning.
By contracting the variance diagram using our method we can numerically detect barren plateaus in ans\"atze significantly faster than traditional approaches based on sampling.

\section{Preliminaries}

\subsection{ZX Calculus}

The ZX calculus is a graphical language for qubit quantum computing \cite{coeckeInteractingQuantumObservables2008, coeckeInteractingQuantumObservables2011}. 
Its diagrams are built from a basic set of generators, which are \textit{Z-spiders} (drawn as green circles), \textit{X-spiders} (drawn as red circles), \textit{Hadamard gates} (drawn as yellow boxes), the identity wire, and crossing wires:
\begin{gather*}
	\input{\tikzitpathfigs/zxintro/zspidercirc.tikz} ~=~ \ket{0}^m\bra{0}^n+e^{i\alpha}\ket{1}^m\bra{1}^n 
	\qquad\qquad
	\input{\tikzitpathfigs/zxintro/xspidercirc.tikz} ~=~ \ket{+}^m\bra{+}^n+e^{i\alpha}\ket{-}^m\bra{-}^n
	\\
	\begin{tikzpicture}
	\begin{pgfonlayer}{nodelayer}
		\node [style=none] (0) at (-0.75, 0) {};
		\node [style=none] (1) at (0.75, 0) {};
		\node [style=hadamard] (2) at (0, 0) {};
	\end{pgfonlayer}
	\begin{pgfonlayer}{edgelayer}
		\draw (0.center) to (1.center);
	\end{pgfonlayer}
\end{tikzpicture}
 ~=~ \ket{+}\bra{0}+\ket{-}\bra{1}
	\qquad
	\begin{tikzpicture}
	\begin{pgfonlayer}{nodelayer}
		\node [style=none] (0) at (-1, -0.125) {};
		\node [style=none] (1) at (1, -0.125) {};
	\end{pgfonlayer}
	\begin{pgfonlayer}{edgelayer}
		\draw (0.center) to (1.center);
	\end{pgfonlayer}
\end{tikzpicture}
 ~=~ \ket{0}\bra{0}+\ket{1}\bra{1}
	\qquad
	\input{\tikzitpathfigs/zxintro/swap.tikz} ~=~ \sum_{i, j=0}^{1}\ket{ji}\bra{ij}
\end{gather*}
Furthermore, we use the following notation for spiders with label 0, wire bending, and the pink spider as a normalised version of the red spider for phases 0 and $\pi$:
\begin{gather*}
	\input{\tikzitpathfigs/zxintro/0phasezdot.tikz}
	\qquad\qquad
	\input{\tikzitpathfigs/zxintro/compactstructures-1.tikz}
	\qquad\qquad
	\input{\tikzitpathfigs/zxintro/compactstructures-2.tikz}
	\qquad\qquad
	\input{\tikzitpathfigs/zxintro/xspidernormalised.tikz} ~:=~ \textstyle{2^{\frac{n+m-2}{2}}} ~ \input{\tikzitpathfigs/zxintro/xspidernormalised-2.tikz}
\end{gather*}
Based on the generators, one can build more complicated diagrams by wiring them together (sequential composition) or putting them next to each other (parallel composition).
For example, common quantum gates can be encoded as ZX diagrams as follows:
\begin{gather*}
	CNOT ~=~ \input{\tikzitpathfigs/zxintro/cnot.tikz}
	\qquad
	CZ ~=~ \input{\tikzitpathfigs/zxintro/cz.tikz} 
	\qquad
	R_Z(\alpha) ~=~ \begin{tikzpicture}
	\begin{pgfonlayer}{nodelayer}
		\node [style=Z phase dot] (0) at (0, 0) {$\alpha$};
		\node [style=none] (1) at (1, 0) {};
		\node [style=none] (2) at (-1, 0) {};
	\end{pgfonlayer}
	\begin{pgfonlayer}{edgelayer}
		\draw (2.center) to (1.center);
	\end{pgfonlayer}
\end{tikzpicture}

	\qquad
	R_X(\alpha) ~=~ \begin{tikzpicture}
	\begin{pgfonlayer}{nodelayer}
		\node [style=red phase dot] (0) at (0, 0) {$\alpha$};
		\node [style=none] (1) at (1, 0) {};
		\node [style=none] (2) at (-1, 0) {};
	\end{pgfonlayer}
	\begin{pgfonlayer}{edgelayer}
		\draw (2.center) to (1.center);
	\end{pgfonlayer}
\end{tikzpicture}

	\qquad
	H ~=~ 
\end{gather*}
ZX diagrams without any inputs and outputs are called \textit{scalar diagrams}, since they represent a single complex number.

An important property of ZX diagrams is the fact that we can arbitrarily deform them topologically by moving the generators around the plane, bending and unbending wires as we go, without changing the interpretation.
This principle is summarised in the slogan \textit{only connectivity matters}.
The main power of the ZX-calculus however comes from its rewrite rules.
While complete sets of rewrite rules exist for the ZX-calculus \cite{hadzihasanovic2018complete,jeandel2018complete,vilmart2019nearminimal}, we only require a small subset for this paper:
\begin{gather*}
	\input{\tikzitpathfigs/zxintro/zfusion.tikz}
	\qquad 
	\input{\tikzitpathfigs/zxintro/colorchange.tikz}
	\qquad 
	\input{\tikzitpathfigs/zxintro/0xcopy.tikz}
	\qquad 
	\input{\tikzitpathfigs/zxintro/vanishtoid.tikz}
\end{gather*}

\subsection{The Triangle Node}
The triangle node was first introduced in \cite{jeandel2018complete} and has since been used as an additional generator for example in \cite{wang2019algebraic,hadzihasanovic2018complete,ng2018completeness}.
In particular, it plays a central role in the ZX$\Delta$ calculus \cite{vilmart2018zx} which was designed for Toffoli+Hadamard quantum mechanics.
We give the interpretation of the triangle, its transpose, and its inverse below:
\[
	\begin{tikzpicture}
	\begin{pgfonlayer}{nodelayer}
		\node [style=none] (0) at (-0.75, 0) {};
		\node [style=none] (1) at (0.75, 0) {};
		\node [style=vtriang] (2) at (0, 0) {};
	\end{pgfonlayer}
	\begin{pgfonlayer}{edgelayer}
		\draw (0.center) to (1.center);
	\end{pgfonlayer}
\end{tikzpicture}
 
	~=~ \begin{pmatrix} 1 & 1 \\ 0 & 1 \end{pmatrix}
	\qquad
	
	~=~ \input{\tikzitpathfigs/zxintro/triangletransp2.tikz}
	~=~ \begin{pmatrix}	1 & 0 \\ 1 & 1 \end{pmatrix}
	\qquad
	\begin{tikzpicture}
	\begin{pgfonlayer}{nodelayer}
		\node [style=none] (0) at (-0.75, 0) {};
		\node [style=none] (1) at (1, 0) {};
		\node [style=vtriang] (2) at (0, 0) {};
		\node [style=none] (3) at (0.25, 0.45) {\scriptsize $-1$};
	\end{pgfonlayer}
	\begin{pgfonlayer}{edgelayer}
		\draw (0.center) to (1.center);
	\end{pgfonlayer}
\end{tikzpicture}
 
	~=~ \input{\tikzitpathfigs/zxintro/triangleinv.tikz} 
	~=~ \begin{pmatrix} 1 & -1 \\ 0 & 1 \end{pmatrix}
\]
The triangle maps $\ket{0}$ to $\ket{0}$ and $\ket{1}$ to $\ket{+}$, or diagrammatically:
\begin{equation*} \label{rule/tri} %\tag{$T$}
	\input{\tikzitpathfigs/zxintro/bas0.tikz}
	\qquad \qquad
	\input{\tikzitpathfigs/zxintro/bas1.tikz}
\end{equation*}
One of the useful feature of the triangle is that it allows us to represent the AND gate, that acts like conjunction on the computational basis:
\[ \input{\tikzitpathfigs/zxintro/andgate.tikz} \]
This makes it very easy to define multi-controlled gates.
For example, a multi-controlled Toffoli gate can be written as
\begin{equation} \label{toffoli}
	\input{\tikzitpathfigs/zxintro/toffoli.tikz} 
\end{equation}
\begin{remark}
	Equation (\ref{tri-T}) implies that any Clifford+Triangle diagram can be represented as a Clifford+T diagram.
	A natural question is whether the opposite is also true, i.e. whether both fragments are equivalent.
	As it turns out, this is not the case (see Corollary \ref{cor:no-magic} in appendix \ref{sec:proofs}).
	However, a single magic state is enough to extend Clifford+Triangle to Clifford+T, since we can use triangles to copy magic states:
	\[ \input{\tikzitpathfigs/zxintro/magicstatecopy.tikz} \]
\end{remark}

\section{Methods}

\subsection{Star Edges}

The typical strategy for optimising ZX diagrams \cite{kissinger2020reducing} starts by reducing them to \textit{graph-like} form, where the diagram only consists of Z-spiders connected by wires with Hadamards on them.
For notational convenience, those \textit{Hadamard edges} are drawn as dashed blue lines:
\[ \begin{tikzpicture}
	\begin{pgfonlayer}{nodelayer}
		\node [style=none] (0) at (-1.25, 0) {};
		\node [style=none] (1) at (1.25, 0) {};
		\node [style=hadamard] (8) at (0, 0) {};
	\end{pgfonlayer}
	\begin{pgfonlayer}{edgelayer}
		\draw (0.center) to (1.center);
	\end{pgfonlayer}
\end{tikzpicture}
 \quad\rightsquigarrow\quad \begin{tikzpicture}
	\begin{pgfonlayer}{nodelayer}
		\node [style=none] (0) at (-1.25, 0) {};
		\node [style=none] (1) at (1.25, 0) {};
	\end{pgfonlayer}
	\begin{pgfonlayer}{edgelayer}
		\draw [style=hadamard edge] (0.center) to (1.center);
	\end{pgfonlayer}
\end{tikzpicture}
 \]
Diagrams can be transformed into graph-like form by turning X-spiders into Z-spiders via Hadamard conjugation and then fusing as many spiders as possible.
Furthermore, parallel edges and self-loops are removed using the following rules:
\begin{equation} \label{zx-hopf} \textstyle{
	\begin{tikzpicture}
	\begin{pgfonlayer}{nodelayer}
		\node [style=Z phase dot] (0) at (0, 0) {$\alpha$};
		\node [style=Z phase dot] (1) at (2, 0) {$\beta$};
		\node [style=none] (2) at (-0.75, 0.625) {};
		\node [style=none] (3) at (-0.75, -0.625) {};
		\node [style=none] (4) at (2.75, 0.625) {};
		\node [style=none] (5) at (2.75, -0.625) {};
		\node [style=none, rotate=90] (6) at (-0.625, 0) {$...$};
		\node [style=none, rotate=90] (7) at (2.625, 0) {$...$};
	\end{pgfonlayer}
	\begin{pgfonlayer}{edgelayer}
		\draw (3.center) to (0);
		\draw (0) to (2.center);
		\draw (5.center) to (1);
		\draw (1) to (4.center);
		\draw [style=star edge, bend left] (0) to (1);
		\draw [style=hadamard edge, bend right] (0) to (1);
	\end{pgfonlayer}
\end{tikzpicture}
 ~=~ \frac{1}{2}~ \begin{tikzpicture}
	\begin{pgfonlayer}{nodelayer}
		\node [style=Z phase dot] (0) at (0, 0) {$\alpha$};
		\node [style=Z phase dot] (1) at (2, 0) {$\beta$};
		\node [style=none] (2) at (-0.75, 0.625) {};
		\node [style=none] (3) at (-0.75, -0.625) {};
		\node [style=none] (4) at (2.75, 0.625) {};
		\node [style=none] (5) at (2.75, -0.625) {};
		\node [style=none, rotate=90] (6) at (-0.625, 0) {$...$};
		\node [style=none, rotate=90] (7) at (2.625, 0) {$...$};
	\end{pgfonlayer}
	\begin{pgfonlayer}{edgelayer}
		\draw (3.center) to (0);
		\draw (0) to (2.center);
		\draw (5.center) to (1);
		\draw (1) to (4.center);
		\draw [style=star edge] (0) to (1);
	\end{pgfonlayer}
\end{tikzpicture}

	\qquad\qquad
	 ~=~ \frac{1}{\sqrt 2}~ }
\end{equation}
The resulting diagrams are called graph-like since they correspond to undirected simple graphs.
However, adding the triangle into the mix complicates the picture somewhat:
since the triangle operation is not symmetric, including triangle-edges would require representing diagrams as directed graphs.
In order to circumvent this, we define a symmetric version of the triangle which we call \textit{star}:
\begin{equation} \tag{$S$} \label{def:star}
	\input{\tikzitpathfigs/star/star.tikz} 
	~:=~ \begin{tikzpicture}
	\begin{pgfonlayer}{nodelayer}
		\node [style=none] (0) at (-0.75, 0) {};
		\node [style=none] (1) at (2, 0) {};
		\node [style=vtriang, rotate=180] (2) at (0, 0) {};
		\node [style=X phase dot] (3) at (1, 0) {$\pi$};
	\end{pgfonlayer}
	\begin{pgfonlayer}{edgelayer}
		\draw (0.center) to (1.center);
	\end{pgfonlayer}
\end{tikzpicture}
 
	~=~ \begin{tikzpicture}
	\begin{pgfonlayer}{nodelayer}
		\node [style=none] (0) at (-1, 0) {};
		\node [style=none] (1) at (1.75, 0) {};
		\node [style=vtriang] (2) at (1, 0) {};
		\node [style=X phase dot] (3) at (0, 0) {$\pi$};
	\end{pgfonlayer}
	\begin{pgfonlayer}{edgelayer}
		\draw (0.center) to (1.center);
	\end{pgfonlayer}
\end{tikzpicture}
 
	~=~ \begin{pmatrix}
			1 & 1 \\
			1 & 0
		\end{pmatrix}
\end{equation}
Note that the star node has the same interpretation as the zero-labelled H-box in the ZH-calculus \cite{backens2018zh} and is indeed symmetric:
\[ \begin{tikzpicture}
	\begin{pgfonlayer}{nodelayer}
		\node [style=none] (0) at (-0.25, 0) {};
		\node [style=none] (1) at (0.25, 0) {};
		\node [style=starh] (2) at (0, 0) {};
		\node [style=none] (3) at (0.75, 0.75) {};
		\node [style=none] (4) at (-0.25, 0.75) {};
		\node [style=none] (5) at (0.25, -0.75) {};
		\node [style=none] (6) at (-0.75, -0.75) {};
	\end{pgfonlayer}
	\begin{pgfonlayer}{edgelayer}
		\draw (0.center) to (1.center);
		\draw (4.center) to (3.center);
		\draw [in=540, out=-180, looseness=1.75] (0.center) to (4.center);
		\draw (6.center) to (5.center);
		\draw [in=360, out=0, looseness=1.50] (5.center) to (1.center);
	\end{pgfonlayer}
\end{tikzpicture}
 ~=~ \input{\tikzitpathfigs/star/star.tikz}  \]
Using the star we can recover all variants of the triangle:
\begin{equation} \label{def:tri}
	\begin{aligned}
	 ~&=~  &
	\qquad\qquad
	 ~&=~ \begin{tikzpicture}
	\begin{pgfonlayer}{nodelayer}
		\node [style=none] (0) at (-0.75, 0) {};
		\node [style=none] (1) at (1.75, 0) {};
		\node [style=starh] (2) at (1, 0) {};
		\node [style=X phase dot] (3) at (0, 0) {$\pi$};
	\end{pgfonlayer}
	\begin{pgfonlayer}{edgelayer}
		\draw (0.center) to (1.center);
	\end{pgfonlayer}
\end{tikzpicture}

	\\
	\begin{tikzpicture}
	\begin{pgfonlayer}{nodelayer}
		\node [style=none] (0) at (-0.75, 0) {};
		\node [style=none] (1) at (1, 0) {};
		\node [style=vtriang, rotate=180] (2) at (0, 0) {};
		\node [style=none] (3) at (0.375, 0.375) {\scriptsize $-1$};
	\end{pgfonlayer}
	\begin{pgfonlayer}{edgelayer}
		\draw (0.center) to (1.center);
	\end{pgfonlayer}
\end{tikzpicture}
 ~&=~ \begin{tikzpicture}
	\begin{pgfonlayer}{nodelayer}
		\node [style=none] (0) at (-1.75, 0) {};
		\node [style=none] (1) at (2.75, 0) {};
		\node [style=starh] (2) at (0, 0) {};
		\node [style=X phase dot] (3) at (1, 0) {$\pi$};
		\node [style=Z phase dot] (4) at (2, 0) {$\pi$};
		\node [style=Z phase dot] (5) at (-1, 0) {$\pi$};
	\end{pgfonlayer}
	\begin{pgfonlayer}{edgelayer}
		\draw (0.center) to (1.center);
	\end{pgfonlayer}
\end{tikzpicture}
 &
	\qquad\qquad
	 ~&=~ \begin{tikzpicture}
	\begin{pgfonlayer}{nodelayer}
		\node [style=none] (0) at (-1.75, 0) {};
		\node [style=none] (1) at (2.75, 0) {};
		\node [style=starh] (2) at (1, 0) {};
		\node [style=X phase dot] (3) at (0, 0) {$\pi$};
		\node [style=Z phase dot] (4) at (2, 0) {$\pi$};
		\node [style=Z phase dot] (5) at (-1, 0) {$\pi$};
	\end{pgfonlayer}
	\begin{pgfonlayer}{edgelayer}
		\draw (0.center) to (1.center);
	\end{pgfonlayer}
\end{tikzpicture}

	\end{aligned}
\end{equation}
As a notational convenience similar to Hadamard edges, we use an orange dashed line to denote a connection between spiders with a star on it:
\[ \begin{tikzpicture}
	\begin{pgfonlayer}{nodelayer}
		\node [style=none] (0) at (-1.25, 0) {};
		\node [style=none] (1) at (1.25, 0) {};
		\node [style=starh] (8) at (0, 0) {};
	\end{pgfonlayer}
	\begin{pgfonlayer}{edgelayer}
		\draw (0.center) to (1.center);
	\end{pgfonlayer}
\end{tikzpicture}
 \quad\rightsquigarrow\quad \begin{tikzpicture}
	\begin{pgfonlayer}{nodelayer}
		\node [style=none] (0) at (0, 0) {};
		\node [style=none] (1) at (2, 0) {};
	\end{pgfonlayer}
	\begin{pgfonlayer}{edgelayer}
		\draw [style=star edge] (0.center) to (1.center);
	\end{pgfonlayer}
\end{tikzpicture}
 \]
We will refer to those as \textit{star edges}.
The following lemma shows how we can remove parallel star edges and star self-loops similar to (\ref{zx-hopf}):
\begin{lemma} \label{lem:hopf}
	Parallel edges and self-loops between Z spiders involving star edges simplify as follows:
	\begin{gather*}\textstyle{
			 ~=~ }
			\qquad\qquad
			 ~=~ \frac{1}{\sqrt 2}~ 
		\\[5pt] \textstyle{
			 ~=~ \frac{1}{\sqrt 2^n}~ }
	\end{gather*}
\end{lemma}
We can now extend the standard notion of graph-like ZX diagrams to diagrams with triangles:
\begin{definition}
	A scalar ZX diagram with triangles is graph-like when all spiders are Z-spiders, they are only connected via Hadamard or star edges, and there are no parallel edges or self-loops.
\end{definition}
\noindent
We can turn any scalar ZX diagram with triangles into graph-like form as follows:
\begin{enumerate}
	\item Replace all triangles with stars according to (\ref{def:tri}).
	
	\item Turn all X-spiders into Z-spiders using Hadamard conjugation.
	
	\item 
	Remove consecutive Hadamards. If there are Hadamards and stars next to each other, insert a green spider in the middle:
	\[ 
		\begin{tikzpicture}
	\begin{pgfonlayer}{nodelayer}
		\node [style=none] (0) at (-0.75, 0) {};
		\node [style=none] (1) at (1.75, 0) {};
		\node [style=starh] (2) at (0, 0) {};
		\node [style=hadamard] (3) at (1, 0) {};
	\end{pgfonlayer}
	\begin{pgfonlayer}{edgelayer}
		\draw (0.center) to (1.center);
	\end{pgfonlayer}
\end{tikzpicture}
 ~=~ \begin{tikzpicture}
	\begin{pgfonlayer}{nodelayer}
		\node [style=none] (0) at (-0.75, 0) {};
		\node [style=none] (1) at (1.75, 0) {};
		\node [style=Z dot] (4) at (0.5, 0) {};
	\end{pgfonlayer}
	\begin{pgfonlayer}{edgelayer}
		\draw [style=star edge] (4) to (0.center);
		\draw [style=hadamard edge] (4) to (1.center);
	\end{pgfonlayer}
\end{tikzpicture}

		\qquad\qquad
		\begin{tikzpicture}
	\begin{pgfonlayer}{nodelayer}
		\node [style=none] (0) at (-0.75, 0) {};
		\node [style=none] (1) at (1.75, 0) {};
		\node [style=starh] (2) at (0, 0) {};
		\node [style=starh] (3) at (1, 0) {};
	\end{pgfonlayer}
	\begin{pgfonlayer}{edgelayer}
		\draw (0.center) to (1.center);
	\end{pgfonlayer}
\end{tikzpicture}
 ~=~ \begin{tikzpicture}
	\begin{pgfonlayer}{nodelayer}
		\node [style=none] (0) at (-0.75, 0) {};
		\node [style=none] (1) at (1.75, 0) {};
		\node [style=Z dot] (4) at (0.5, 0) {};
	\end{pgfonlayer}
	\begin{pgfonlayer}{edgelayer}
		\draw [style=star edge] (4) to (0.center);
		\draw [style=star edge] (4) to (1.center);
	\end{pgfonlayer}
\end{tikzpicture}

	\]
	
	\item Remove parallel edges and self-loops using (\ref{zx-hopf}) and Lemma \ref{lem:hopf}.
\end{enumerate}

\subsection{Stabiliser Decompositions}

After turning diagrams into graph-like form, we decompose star edges into sums of stabiliser diagrams.
The naive decomposition turns a single star edge into two terms:
\begin{equation} \label{decomp/star-1} \tag{$D_1$}
	\input{\tikzitpathfigs/decomps/star-1/star.tikz} ~=~ \textstyle{
	\sqrt 2~~  ~~+~~ 
	2~~ }
\end{equation}
Using this decomposition (\ref{decomp/star-1}), we can represent a diagram with $t$ triangles via $2^{\beta t}$ terms with $\beta = 1$.
However, similar to magic states it turns out that there are more efficient decompositions that yield a lower $\beta$ when looking at multiple star edges tensored together:
\begin{align}
	\label{decomp/star-2} \tag{$D_2$}
	\input{\tikzitpathfigs/decomps/star-2/star.tikz} ~&=~ \textstyle{
		\frac{1}{\sqrt 2}~  ~~+~~
		\frac{1}{\sqrt 2}~  ~~+~ ~
		4~~ \input{\tikzitpathfigs/decomps/star-2/3.tikz}}
	\\[5pt]
	\label{decomp/star-3} \tag{$D_3$}
	\input{\tikzitpathfigs/decomps/star-3/star.tikz} ~&=~ \textstyle{
		\frac{1}{2\sqrt 2}~~  ~~+~~
		\frac{1}{2\sqrt 2}~~  ~~+~~
		\frac{1}{2\sqrt 2}~~ \input{\tikzitpathfigs/decomps/star-3/3.tikz} ~+~~
		\frac{1}{\sqrt 2}~~ \input{\tikzitpathfigs/decomps/star-3/4.tikz} ~~+~~
		8~~ \input{\tikzitpathfigs/decomps/star-3/5.tikz}}
\end{align}
We note that equivalent decompositions to (\ref{decomp/star-2}) and (\ref{decomp/star-3}) were previously found in \cite{laakkonen2022graphical} while studying counting problems using the ZH calculus.
They yield $\beta \approx 0.792$ and $\beta \approx 0.774$, respectively.
Furthermore, we find the following decompositions involving tensors of star edges connected to Z spiders:
\begin{align}
	\label{decomp/star-3-state/0} \tag{$D_4$}
	\input{\tikzitpathfigs/decomps/star-3-state/star-0.tikz} ~&=~ \textstyle{
	3~~  ~~-~~
	 ~~+~~
	\frac{3}{\sqrt{2}}~~ \input{\tikzitpathfigs/decomps/star-3-state/3.tikz} ~~-~~
	\frac{3}{2\sqrt{2}}~~ \input{\tikzitpathfigs/decomps/star-3-state/4.tikz}}
	\\[5pt]
	\label{decomp/star-3-state/pi2} \tag{$D_5$}
	\input{\tikzitpathfigs/decomps/star-3-state/star-pi2.tikz} ~&=~ \textstyle{
	\frac{1 \pm 3i}{2}~~  ~~+~~
	\frac{1 \mp i}{2}~~  ~~-~~
	\frac{3-i}{2\sqrt{2}}~~ \input{\tikzitpathfigs/decomps/star-3-state/3.tikz} ~~+~~
	\frac{1 \mp i}{2\sqrt{2}}~~ \input{\tikzitpathfigs/decomps/star-3-state/4.tikz}}
\end{align} 
Both decompositions (\ref{decomp/star-3-state/0}) and (\ref{decomp/star-3-state/pi2}) correspond to $\beta = \frac{2}{3} \approx 0.667$.
Note that we do not consider a version of this decomposition where the spiders have a $\pi$ phase since this already simplifies to a Clifford state (see equation (\ref{lem:simp/state}) in the next section).
Finally, we apply a particularly efficient decomposition that applies when multiple star edges are connected to the same Z spider:
\begin{lemma} \label{lem:decomp-split}
	Multiple star edges connected to the same spider in a graph-like diagram can always be decomposed into two stabiliser terms:
	\begin{equation} \textstyle{
	\label{decomp/split} \tag{$D_6$}
	 
	~=~~ \frac{1}{\sqrt 2^{n}}~ \begin{tikzpicture}
	\begin{pgfonlayer}{nodelayer}
		\node [style=Z phase dot] (0) at (0.75, 1) {$\gamma_1$};
		\node [style=none] (1) at (1.5, 1.5) {};
		\node [style=none] (2) at (1.5, 0.5) {};
		\node [style=none, rotate=90] (3) at (1.375, 1) {$...$};
		\node [style=Z phase dot] (4) at (0.75, -1) {$\gamma_m$};
		\node [style=none] (5) at (1.5, -0.5) {};
		\node [style=none] (6) at (1.5, -1.5) {};
		\node [style=none, rotate=90] (7) at (1.375, -1) {$...$};
		\node [style=none, rotate=90] (8) at (0.75, 0) {$...$};
		\node [style=Z phase dot] (9) at (-0.75, 1) {$\beta_1$};
		\node [style=none] (10) at (-1.5, 1.5) {};
		\node [style=none] (11) at (-1.5, 0.5) {};
		\node [style=none, rotate=90] (12) at (-1.375, 1) {$...$};
		\node [style=Z phase dot] (13) at (-0.75, -1) {$\beta_n$};
		\node [style=none] (14) at (-1.5, -0.5) {};
		\node [style=none] (15) at (-1.5, -1.5) {};
		\node [style=none, rotate=90] (16) at (-1.375, -1) {$...$};
		\node [style=none, rotate=90] (17) at (-0.75, 0) {$...$};
	\end{pgfonlayer}
	\begin{pgfonlayer}{edgelayer}
		\draw (2.center) to (0);
		\draw (0) to (1.center);
		\draw (6.center) to (4);
		\draw (4) to (5.center);
		\draw (11.center) to (9);
		\draw (9) to (10.center);
		\draw (15.center) to (13);
		\draw (13) to (14.center);
	\end{pgfonlayer}
\end{tikzpicture}
 ~~+~~ \frac{e^{i\alpha}}{\sqrt 2^{n+m}}~ \begin{tikzpicture}
	\begin{pgfonlayer}{nodelayer}
		\node [style=Z phase dot] (0) at (2, 1) {$\gamma_1$};
		\node [style=none] (1) at (2.75, 1.5) {};
		\node [style=none] (2) at (2.75, 0.5) {};
		\node [style=none, rotate=90] (3) at (2.625, 1) {$...$};
		\node [style=Z phase dot] (4) at (2, -1) {$\gamma_m$};
		\node [style=none] (5) at (2.75, -0.5) {};
		\node [style=none] (6) at (2.75, -1.5) {};
		\node [style=none, rotate=90] (7) at (2.625, -1) {$...$};
		\node [style=none, rotate=90] (8) at (2, 0) {$...$};
		\node [style=Z phase dot] (9) at (-1, 1) {$\beta_1+\pi$};
		\node [style=none] (10) at (-2.25, 1.5) {};
		\node [style=none] (11) at (-2.25, 0.5) {};
		\node [style=none, rotate=90] (12) at (-2.125, 1) {$...$};
		\node [style=Z phase dot] (13) at (-1, -1) {$\beta_n+\pi$};
		\node [style=none] (14) at (-2.25, -0.5) {};
		\node [style=none] (15) at (-2.25, -1.5) {};
		\node [style=none, rotate=90] (16) at (-2.125, -1) {$...$};
		\node [style=none, rotate=90] (17) at (-1, 0) {$...$};
		\node [style=Z dot] (20) at (0.75, 1) {};
		\node [style=Z dot] (21) at (0.75, -1) {};
	\end{pgfonlayer}
	\begin{pgfonlayer}{edgelayer}
		\draw (2.center) to (0);
		\draw (0) to (1.center);
		\draw (6.center) to (4);
		\draw (4) to (5.center);
		\draw (11.center) to (9);
		\draw (9) to (10.center);
		\draw (15.center) to (13);
		\draw (13) to (14.center);
		\draw [style=hadamard edge] (21) to (4);
		\draw [style=hadamard edge] (0) to (20);
	\end{pgfonlayer}
\end{tikzpicture}
}
	\end{equation}
\end{lemma}
\noindent
For $m$ star edges connect to the same spider, (\ref{decomp/split}) yields $\beta = \frac{1}{m}$.
This means that this decomposition performs better than (\ref{decomp/star-1}) - (\ref{decomp/star-3-state/pi2}) as soon as $m > 1$.

\subsection{Simplifying Diagrams with Star Edges} \label{sec:simplify}

After turning the diagram into graph-like form and after each decomposition is applied, we try to simplify the diagram in order to reduce the number of star edges.
The stabiliser decomposition approach is particularly effective when simulating circuits dominated by Cliffords \cite{bravyi2016improved}, i.e. the star edges occur sparsely within the graph.
In those cases, we expect to have clusters inside the graph that are only connected via Hadamard edges.
Thus, wherever possible we first apply the standard ZX simplification techniques from \cite{kissinger2020reducing}, removing spiders with phases $\pm\frac{\pi}{2}$ and pairs of spiders with phases $0$ or $\pi$ using local complementations and pivoting.
If the diagram contains Ts, we furthermore apply the T-count reduction techniques from \cite{kissinger2020reducing}.

Next, we discuss simplifications that reduce the number of star edges.
First, if we have a $\pi$ spider with a single star edge, then we can replace the star with a Hadamard edge:
\medskip
\begin{equation}\label{lem:simp/state}
	\begin{tikzpicture}
	\begin{pgfonlayer}{nodelayer}
		\node [style=Z phase dot] (0) at (0, 0) {$\pi$};
		\node [style=none] (1) at (1.5, 0) {};
	\end{pgfonlayer}
	\begin{pgfonlayer}{edgelayer}
		\draw [style=star edge] (1.center) to (0);
	\end{pgfonlayer}
\end{tikzpicture}
 ~=~ \frac{1}{\sqrt 2}~ \begin{tikzpicture}
	\begin{pgfonlayer}{nodelayer}
		\node [style=Z phase dot] (0) at (0, 0) {$\pi$};
		\node [style=none] (1) at (1.5, 0) {};
	\end{pgfonlayer}
	\begin{pgfonlayer}{edgelayer}
		\draw [style=hadamard edge] (1.center) to (0);
	\end{pgfonlayer}
\end{tikzpicture}

\end{equation}
\medskip
We observe that this pattern sometimes occurs as a result of Clifford simplifications when other Hadamard edges are toggled.
Furthermore, we can simplify the following pattern:
\medskip
\begin{equation} \label{lem:simp/star-pi-star}
	\begin{tikzpicture}
	\begin{pgfonlayer}{nodelayer}
		\node [style=Z phase dot] (0) at (0, 0) {$\pi$};
		\node [style=none] (1) at (-1.25, 0) {};
		\node [style=none] (2) at (1.25, 0) {};
	\end{pgfonlayer}
	\begin{pgfonlayer}{edgelayer}
		\draw [style=star edge] (1.center) to (2.center);
	\end{pgfonlayer}
\end{tikzpicture}
 ~=~ \input{\tikzitpathfigs/simp/star-pi-star-2.tikz}
\end{equation}
\medskip
Note that there are many more graph-like patterns in which the number of star edges can be reduced. For example:
\[
	\input{\tikzitpathfigs/simp/had-1.tikz} ~=~ \textstyle{\sqrt{2}} ~ \begin{tikzpicture}
	\begin{pgfonlayer}{nodelayer}
		\node [style=Z phase dot] (0) at (-0.75, 0) {$\pi$};
		\node [style=Z phase dot] (1) at (0.75, 0) {$\pi$};
		\node [style=none] (2) at (1.75, 0) {};
		\node [style=none] (3) at (-1.75, 0) {};
	\end{pgfonlayer}
	\begin{pgfonlayer}{edgelayer}
		\draw [style=hadamard edge] (3.center) to (2.center);
	\end{pgfonlayer}
\end{tikzpicture}

	\qquad\qquad
	\input{\tikzitpathfigs/simp/had-star-1.tikz} ~=~ \textstyle{\sqrt{2}} ~ \input{\tikzitpathfigs/simp/had-star-2.tikz}
\]
However, empirically these pattern only seem to occur rarely during contraction.
Therefore, we opt to not use them for our numerical experiments to not incur the additional cost of searching for them.
Our most effective simplification strategy is based on the observation that many of the stabiliser decompositions we consider introduce terms containing \begin{tikzpicture}
	\begin{pgfonlayer}{nodelayer}
		\node [style=Z dot] (0) at (0, 0) {};
		\node [style=none] (1) at (0.75, 0) {};
	\end{pgfonlayer}
	\begin{pgfonlayer}{edgelayer}
		\draw [style=hadamard edge] (1.center) to (0);
	\end{pgfonlayer}
\end{tikzpicture}
 or \begin{tikzpicture}
	\begin{pgfonlayer}{nodelayer}
		\node [style=Z phase dot] (0) at (0, 0) {$\pi$};
		\node [style=none] (1) at (1, 0) {};
	\end{pgfonlayer}
	\begin{pgfonlayer}{edgelayer}
		\draw [style=hadamard edge] (1.center) to (0);
	\end{pgfonlayer}
\end{tikzpicture}
.
When connected to a Z spider, those states are copied:
\[ 
	\textstyle{\input{\tikzitpathfigs/simp/copy-explain-1.tikz} ~=~ \frac{1}{\sqrt{2}^{n-1}} ~ \input{\tikzitpathfigs/simp/copy-explain-2.tikz} 
	\qquad\qquad
	\input{\tikzitpathfigs/simp/copy-explain-3.tikz} ~=~ \frac{1}{\sqrt{2}^{n-1}} ~ \input{\tikzitpathfigs/simp/copy-explain-4.tikz}}
\]
Using this principle we obtain the following simplification strategy:
\begin{lemma} \label{lem:simp/copy}
	For all $\alpha,\beta_1,...,\beta_n,\gamma_1,...,\gamma_m \in \mathbb R$ we have
	\[\textstyle{
		\input{\tikzitpathfigs/simp/copy-1.tikz} ~=~ \frac{1}{\sqrt 2^{n-1}}~ \input{\tikzitpathfigs/simp/copy-3.tikz}
		\qquad\qquad
		\input{\tikzitpathfigs/simp/copy-pi-1.tikz} ~=~ \frac{e^{i\alpha}}{\sqrt 2^{n+m-1}}~ \input{\tikzitpathfigs/simp/copy-pi-3.tikz}
	}\]
\end{lemma}
Note that after applying the second equation of Lemma \ref{lem:simp/copy}, we can immediately apply the first equation for each of the spiders labelled with $\gamma_1,...,\gamma_m$, triggering a cascading chain of simplifications.
Thus, applying Lemma \ref{lem:simp/copy} exhaustively will always decrease the number of spiders.
Furthermore, as explained above, this simplification is guaranteed to apply for many of our stabiliser decompositions (for example six times for the last term of (\ref{decomp/star-3})).

\subsection{Full Algorithm}

Below we summarise all steps of our contraction algorithm:
\begin{enumerate}
	\item
	Given a scalar ZX diagram with triangles, turn it into graph-like form.
	
	\item \label{algo-simp}
	Simplify the diagram according to Section \ref{sec:simplify}.
	
	\item  \label{algo-decomp}
	Determine which decompositions from (\ref{decomp/star-1}) - (\ref{decomp/split}) are applicable.
	If the diagram contains T-spiders, also search for decompositions from \cite{kissinger2022classical}.
	Among the candidates, greedily apply the decomposition with the lowest $\alpha$ or $\beta$.

	\item 
	For each term, repeat steps \ref{algo-simp} and \ref{algo-decomp} recursively until no star edges and T-spiders are left, in which case the diagram can be efficiently contracted.
	
	\item 
	The sum of all terms gives the scalar value of the initial diagram.
\end{enumerate}
In the worst case, the contraction requires $2^{\alpha n + \beta m}$ terms where $n,m$ are the number of T-spiders and triangles, respectively, and $\alpha\approx 0.396$, $\beta\approx 0.774$.

\section{Classical Simulation}

tion we apply our stabiliser decomposition approach to classical simulation of quantum circuits involving multi-controlled gates.
All experiments in this section are run on a consumer laptop using a single CPU core.

\subsection{Random Circuits}

\begin{figure}
	\centering
	\begin{subfigure}[b]{0.49\textwidth}
		\centering
		\includegraphics[scale=0.7]{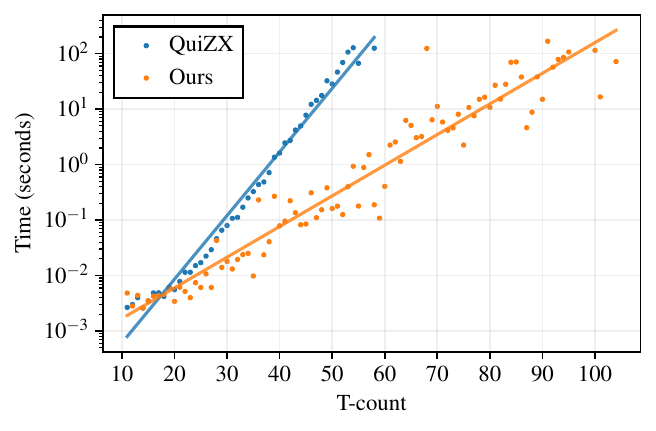}
	\end{subfigure}
	\hfill
	\begin{subfigure}[b]{0.49\textwidth}
		\centering
		\includegraphics[scale=0.7]{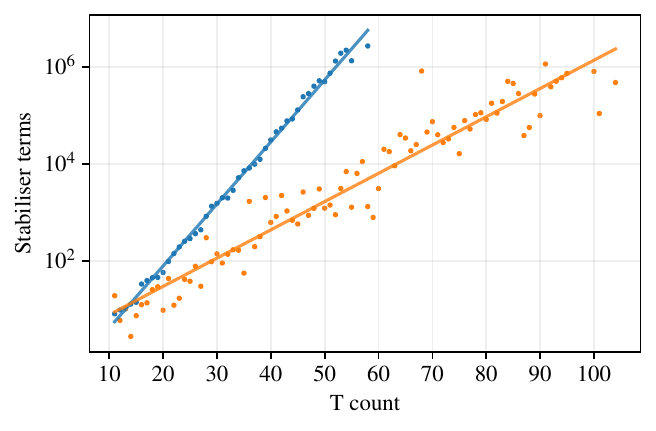}
	\end{subfigure}
	\caption{Runtime and total number of stabiliser terms for random 50-qubit Clifford+T+CCZ simulations sorted by T-count of the QuiZX circuit after initial simplification. Averaged over all circuits with the same T-count. Note that the y-axis is log-scaled.}
	\label{fig:random-time}
\end{figure}

We generate 500 random 50 qubit circuits consisting of T, CCZ and Clifford-gates (CNOT, CZ, Hadamard, and S).
Concretely, we randomly sample 500-800 gates from this set and place them in the circuit, where T and CCZ gates are sampled with probability 5\% each.
We calculate the amplitude of these circuits for the fixed input $\ket{+}^{\otimes 50}$ and output $\bra{+}^{\otimes 50}$ following the benchmarking approach in \cite{kissinger2022classical}.
We compare the performance of our method against the implementation of \cite{kissinger2022classical} in QuiZX which represents CCZ gates using seven T-spiders:
\begin{equation} \label{ccz-T}
	\begin{tikzpicture}
	\begin{pgfonlayer}{nodelayer}
		\node [style=none] (0) at (-0.75, 0) {};
		\node [style=none] (1) at (0.75, 0) {};
		\node [style=none] (2) at (-0.75, 1.25) {};
		\node [style=none] (3) at (0.75, 1.25) {};
		\node [style=none] (4) at (0.75, -1.25) {};
		\node [style=none] (5) at (-0.75, -1.25) {};
		\node [style=cnot ctrl] (6) at (0, 1.25) {};
		\node [style=cnot ctrl] (7) at (0, 0) {};
		\node [style=cnot ctrl] (8) at (0, -1.25) {};
	\end{pgfonlayer}
	\begin{pgfonlayer}{edgelayer}
		\draw (2.center) to (3.center);
		\draw (1.center) to (0.center);
		\draw (5.center) to (4.center);
		\draw (8) to (6);
	\end{pgfonlayer}
\end{tikzpicture}
 ~=~ \textstyle{4\sqrt{2}} ~~ \begin{tikzpicture}
	\begin{pgfonlayer}{nodelayer}
		\node [style=none] (0) at (-1.25, 0) {};
		\node [style=none] (1) at (1.5, 0) {};
		\node [style=none] (2) at (-1.25, 2) {};
		\node [style=none] (3) at (1.5, 2) {};
		\node [style=none] (4) at (1.5, -2) {};
		\node [style=none] (5) at (-1.25, -2) {};
		\node [style=Z phase dot] (6) at (0, -2) {$\frac{\pi}{4}$};
		\node [style=Z phase dot] (7) at (0, 0) {$\frac{\pi}{4}$};
		\node [style=Z phase dot] (8) at (0, 2) {$\frac{\pi}{4}$};
		\node [style=Z dot] (9) at (-1.125, -1) {};
		\node [style=Z dot] (10) at (1.125, -1) {};
		\node [style=Z dot] (11) at (-1.125, 1) {};
		\node [style=Z dot] (12) at (1.125, 1) {};
		\node [style=Z phase dot] (13) at (-2.375, 1) {$-\frac{\pi}{4}$};
		\node [style=Z phase dot] (14) at (-2.375, -1) {$-\frac{\pi}{4}$};
		\node [style=Z phase dot] (15) at (2.375, 1) {$-\frac{\pi}{4}$};
		\node [style=Z phase dot] (16) at (2.375, -1) {$-\frac{\pi}{4}$};
	\end{pgfonlayer}
	\begin{pgfonlayer}{edgelayer}
		\draw (2.center) to (3.center);
		\draw (1.center) to (0.center);
		\draw (5.center) to (4.center);
		\draw [style=hadamard edge] (6) to (9);
		\draw [style=hadamard edge] (9) to (7);
		\draw [style=hadamard edge] (7) to (10);
		\draw [style=hadamard edge] (10) to (6);
		\draw [style=hadamard edge] (11) to (7);
		\draw [style=hadamard edge] (7) to (12);
		\draw [style=hadamard edge] (12) to (8);
		\draw [style=hadamard edge] (8) to (11);
		\draw [style=hadamard edge] (11) to (13);
		\draw [style=hadamard edge] (14) to (9);
		\draw [style=hadamard edge] (10) to (16);
		\draw [style=hadamard edge] (12) to (15);
		\draw [style=hadamard edge] (9) to (8);
		\draw [style=hadamard edge] (6) to (12);
	\end{pgfonlayer}
\end{tikzpicture}

\end{equation}
For our method, we use the following decomposition of the CCZ gate that only requires two triangles~\cite{jeandel2018complete}:
\begin{equation} \label{ccz-tri}
	 ~=~ \textstyle{2} ~~ \begin{tikzpicture}
	\begin{pgfonlayer}{nodelayer}
		\node [style=none] (0) at (-1.5, -1) {};
		\node [style=none] (1) at (1.5, -1) {};
		\node [style=none] (2) at (-1.5, 2) {};
		\node [style=none] (3) at (1.5, 2) {};
		\node [style=none] (4) at (1.5, -2) {};
		\node [style=none] (5) at (-1.5, -2) {};
		\node [style=Z dot] (6) at (0, 2) {};
		\node [style=Z dot] (7) at (-0.75, -1) {};
		\node [style=Z dot] (8) at (0.75, -2) {};
		\node [style=Z dot] (9) at (0.75, 0) {};
		\node [style=Z phase dot] (10) at (-0.75, 0) {$\pi$};
		\node [style=dtriangle] (11) at (-0.75, 1.125) {};
		\node [style=dtriangle, rotate=180] (12) at (0.75, 1) {};
	\end{pgfonlayer}
	\begin{pgfonlayer}{edgelayer}
		\draw (2.center) to (3.center);
		\draw (1.center) to (0.center);
		\draw (5.center) to (4.center);
		\draw [style=hadamard edge] (8) to (9);
		\draw [style=hadamard edge] (7) to (10);
		\draw [in=-165, out=90] (11) to (6);
		\draw [in=90, out=-15] (6) to (12);
		\draw [style=hadamard edge] (10) to (11);
		\draw (9) to (12);
		\draw [style=hadamard edge] (10) to (9);
	\end{pgfonlayer}
\end{tikzpicture}

\end{equation}
Figure \ref{fig:random-time} shows the runtime and total number of stabiliser terms of our method compared to \cite{kissinger2022classical} on the 500 benchmarking circuits sorted by T-count.
Note that by virtue of being randomly sampled, many of these circuits can be simplified, reducing the T-count.
Thus, we report the T-counts in Figure \ref{fig:random-time} \textit{after} initially applying the T reduction techniques from \cite{kissinger2020reducing}.
Furthermore, we impose a timeout of five minutes after which simulations are aborted.
Using an exponential fit, we find a term growth rate of $2^{0.42t}$ for QuiZX compared to $2^{0.19t}$ using our method depending on the number of Ts $t$.
At T-counts between 50 and 60 we already observe a 1-2 orders of magnitude speed-up using our method.
While QuiZX was not able to simulate circuits with T-counts greater than 60 within our time limit, we were able to go up to more than 100 Ts.

\subsection{Modified Hidden Shift Circuits}

\begin{figure}
	\centering
	\begin{minipage}[t]{0.3\textwidth}
		\centering
		\includegraphics[scale=0.7]{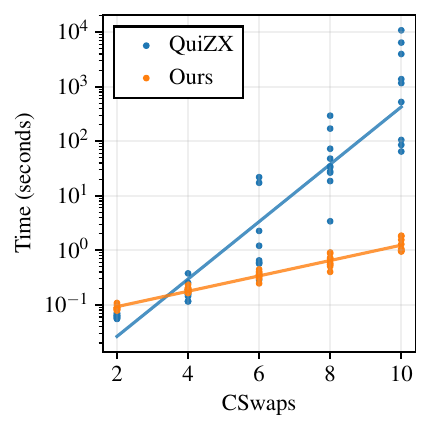}
		\caption{Simulation times for random modified hidden shift circuits with increasing number of controlled swaps.}
		\label{fig:hidden-shift-time}
	\end{minipage}
	\hfill
	\begin{minipage}[t]{0.6\textwidth}
		\centering
		\includegraphics[scale=0.7]{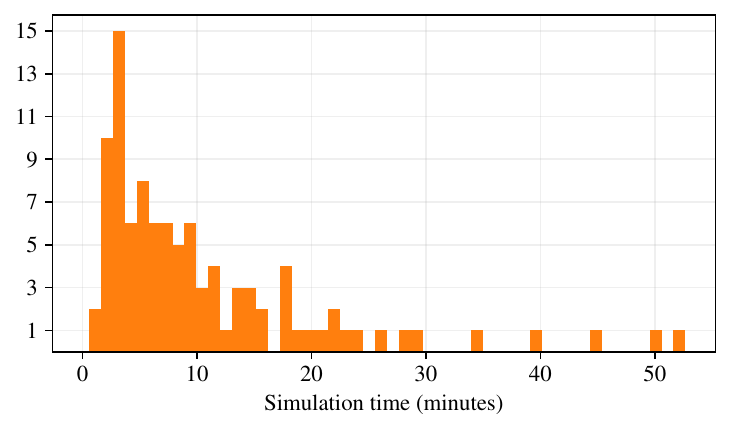}
		\caption{Running time distribution for simulating 100 random modified hidden shift circuits with 30 controlled swaps (equivalent to T-count 210) using our method.}
		\label{fig:hidden-shift-distribution}
	\end{minipage}
\end{figure}

The family of hidden shift circuits for benchmarking stabiliser decomposition based simulators was first introduced in \cite{bravyi2016improved} and has since been used for benchmarking in \cite{bravyi2019simulation,kissinger2022simulating,kissinger2022classical,peres2022quantum}.
However, as previously pointed out in \cite{codsi2022cutting}, a recent improvement in the T-count reduction algorithm of QuiZX seems to completely trivialise the simulation of these circuits:
The initial simplification step usually removes \textit{all} T gates, leaving us with a Clifford circuit that can be easily simulated.
Therefore, we propose a modified version of the hidden shift circuits described in \cite{bravyi2016improved} where the random oracles contain pairs of controlled swaps (Fredkin gates) instead of CCZ gates.
Empirically, this change makes the circuits significantly harder to simulate as demonstrated by our experiments.
We implement the controlled swap gates as follows
\[
	\begin{tikzpicture}
	\begin{pgfonlayer}{nodelayer}
		\node [style=none] (0) at (-0.75, 0) {};
		\node [style=none] (1) at (0.75, 0) {};
		\node [style=none] (2) at (-0.75, 1.25) {};
		\node [style=none] (3) at (0.75, 1.25) {};
		\node [style=none] (4) at (0.75, -1.25) {};
		\node [style=none] (5) at (-0.75, -1.25) {};
		\node [style=cnot ctrl] (6) at (0, 1.25) {};
		\node [style=none] (7) at (0, -1.25) {};
		\node [style=none] (8) at (-0.2, 0.2) {};
		\node [style=none] (9) at (0.2, -0.2) {};
		\node [style=none] (10) at (0.2, 0.2) {};
		\node [style=none] (11) at (-0.2, -0.2) {};
		\node [style=none] (12) at (-0.2, -1.05) {};
		\node [style=none] (13) at (0.2, -1.45) {};
		\node [style=none] (14) at (0.2, -1.05) {};
		\node [style=none] (15) at (-0.2, -1.45) {};
	\end{pgfonlayer}
	\begin{pgfonlayer}{edgelayer}
		\draw (2.center) to (3.center);
		\draw (1.center) to (0.center);
		\draw (5.center) to (4.center);
		\draw (7.center) to (6);
		\draw (8.center) to (9.center);
		\draw (11.center) to (10.center);
		\draw (12.center) to (13.center);
		\draw (15.center) to (14.center);
	\end{pgfonlayer}
\end{tikzpicture}
 ~=~ \begin{tikzpicture}
	\begin{pgfonlayer}{nodelayer}
		\node [style=none] (0) at (-3.25, 0) {};
		\node [style=none] (1) at (3.25, 0) {};
		\node [style=none] (2) at (-3.25, 1.25) {};
		\node [style=none] (3) at (3.25, 1.25) {};
		\node [style=none] (4) at (3.25, -1.25) {};
		\node [style=none] (5) at (-3.25, -1.25) {};
		\node [style=cnot ctrl] (6) at (0, 1.25) {};
		\node [style=cnot ctrl] (7) at (0, 0) {};
		\node [style=cnot ctrl] (8) at (0, -1.25) {};
		\node [style=cnot targ] (9) at (-2.5, 0) {};
		\node [style=cnot targ] (10) at (2.5, 0) {};
		\node [style=cnot ctrl] (11) at (-2.5, -1.25) {};
		\node [style=cnot ctrl] (12) at (2.5, -1.25) {};
		\node [style=gate] (13) at (-1.25, -1.25) {$H$};
		\node [style=gate] (14) at (1.25, -1.25) {$H$};
	\end{pgfonlayer}
	\begin{pgfonlayer}{edgelayer}
		\draw (2.center) to (3.center);
		\draw (1.center) to (0.center);
		\draw (5.center) to (4.center);
		\draw (8) to (6);
		\draw (12) to (10);
		\draw (9) to (11);
	\end{pgfonlayer}
\end{tikzpicture}

\]
where we decompose the CCZ according to (\ref{ccz-T}) and (\ref{ccz-tri}), respectively.
We simulate modified hidden shift circuits with up to 10 controlled swaps (equivalent to T-count 70).
Figure \ref{fig:hidden-shift-time} plots the runtime of our method compared to QuiZX.
Using an exponential fit, we observe a growth rate of $2^{1.74s}$ for QuiZX compared to $2^{0.47s}$ using our method depending on the number of controlled swaps $s$.
For the largest circuits we perform simulation 2-4 orders of magnitude faster.
In one instance QuiZX required almost 3 hours to sample from a single circuit, whereas our method only took about 1.5 seconds.

Next, we simulate 100 modified hidden shift circuits with 30 controlled swaps (equivalent to T-count 210) using our method.
The runtime distribution over those runs is shown in Figure \ref{fig:hidden-shift-distribution}.
We were able to simulate 64\% of circuits within 10 minutes and 95\% of circuits within 30 minutes.
The longest simulation time we observed was 52 minutes.

\section{Numerical Barren Plateau Detection}

In this section we describe an application of our method in the context quantum machine learning.
A common challenge when training parametrised quantum circuits using gradient-based optimisation methods is the \textit{barren plateau phenomenon} \cite{mcclean2018barren}.
Roughly, it describes the problem that the gradient landscape of many quantum circuits flattens exponentially with increasing circuit sizes, making gradient-descent on such circuits increasingly difficult.
For a parametrised unitary $U(\vec\theta)$ on $n$ qubits and a Hermitian observable $H$ we denote the corresponding expectation value as $\langle H\rangle = \bra{0}U^\dagger(\vec\theta)HU(\vec\theta)\ket{0}$.
If we assume that the parameters $\vec\theta$ are uniformly sampled from the interval $[-\pi,\pi]$, one can show that $\mathbb E\left(\partial \langle H\rangle / \partial \theta_i\right) = 0$ for a wide class of circuits.
We say barren plateaus are present if $\text{Var}\left(\partial \langle H\rangle / \partial \theta_i\right)$ vanishes exponentially as a function of the number of qubits $n$.
The probability that the gradient $\partial \langle H\rangle / \partial \theta_i$ is non-zero up to some fixed precision is exponentially small in $n$ in that case.
As a result, the sampling complexity of estimating the circuit gradient is exponential in the number of qubits, making the gradient optimisation of such circuits intractable.

Since barren plateaus are a major obstacle to successfully optimising parametrised quantum circuits, detecting their presence or absence in ans\"atze is of great interest.
While the ZX calculus has previously been used towards this end by analytically studying the gradient variance \cite{zhao2021analyzingbarren, martin2023barren, wang2022differentiating, koch2022quantum}, we focus on numerical methods for barren plateau detection by computing $\text{Var}\left(\partial \langle H\rangle / \partial \theta_i\right)$ for increasing $n$.

\subsection{Diagrammatic Variance Calculation}

Zhao and Gao \cite{zhao2021analyzingbarren} were the first to employ the ZX-calculus to analyse barren plateaus.
They express $\text{Var}\left(\partial \langle H\rangle / \partial \theta_i\right)$ as a linear combination of diagrams with an exponential number of terms, which they handle using tensor networks.
We improved on this in \cite{wang2022differentiating} by expressing the variance in a single diagram, allowing the analysis of barren plateaus to be carried out entirely within the framework of ZX.
Assuming that each parameter appears exactly once in the ZX diagram representing an ansatz, we use the following notation for diagrams representing the expectation values:
\[
	\langle H\rangle 
	~=~ \input{\tikzitpathfigs/barren/exp-1.tikz} 
	~:=~ \input{\tikzitpathfigs/barren/exp-2.tikz} 
	~=~ \bra{0}U^\dagger(\vec\theta)HU(\vec\theta)\ket{0}
\]
We refer to \cite{shaikh2022sum} on how to represent $H$ as a diagram.
The gradient variance can then be computed as follows \cite{wang2022differentiating}:
\begin{equation} \label{variance}
	\operatorname{Var}\left(\frac{\partial\langle H\rangle}{\partial \theta_j}\right) ~=~ \input{\tikzitpathfigs/barren/variance.tikz} 
\end{equation}
The ZX diagrams of most quantum ans\"atze only contain spiders with phases that are either parameters or multiples of $\pi/2$. 
Since the parametrised spiders have been removed in (\ref{variance}), it is a Clifford+Triangle diagram and can be evaluated using our stabiliser decomposition algorithm involving $2^{\beta(n-1)}$ terms for a circuit with $n$ parameters.
In comparison, Zhao and Gao's method \cite{zhao2021analyzingbarren} would require $3^{n-1}$ terms.

\subsection{Examples} \label{sec:examples}

\begin{figure}
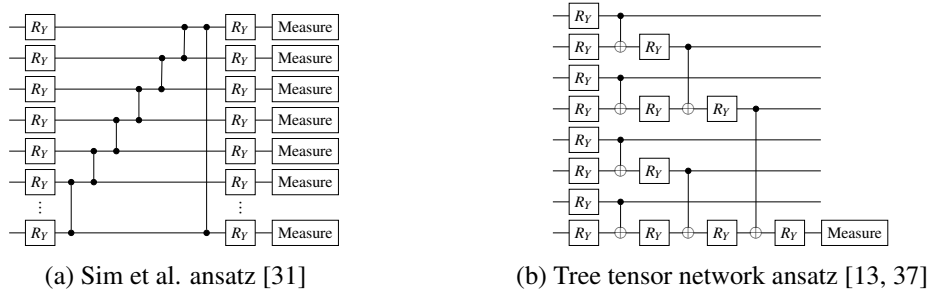

	\centering
	\begin{subfigure}[b]{0.4\textwidth}
		\centering
		\scalebox{0.6}{\input{\tikzitpathfigs/barren/sim10.tikz}}
		\caption{Sim et al. ansatz \cite{sim2019expressibility}}
		\label{fig:sim10}
	\end{subfigure}
	\qquad
	\begin{subfigure}[b]{0.4\textwidth}
		\centering
		\scalebox{0.6}{\input{\tikzitpathfigs/barren/tree.tikz}}
		\caption{Tree tensor network ansatz \cite{grant2018hierarchical,zhao2021analyzingbarren}}
		\label{fig:tree}
	\end{subfigure}
	\caption{Example ans\"atze for numerical barren plateau detection.}
	\label{fig:ansatze}
\end{figure}

\begin{figure}
	\centering
	\begin{subfigure}[b]{0.39\textwidth}
		\centering
		$\includegraphics[scale=0.7]{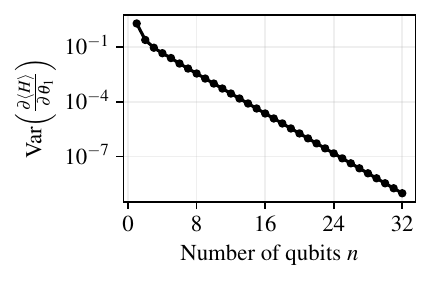}$
		\caption{Sim et al. ansatz \cite{sim2019expressibility} with $H = Z^{\otimes n}$.}
		\label{fig:sim10-variance}
	\end{subfigure}
	\begin{subfigure}[b]{0.6\textwidth}
		\centering
		\includegraphics[scale=0.7]{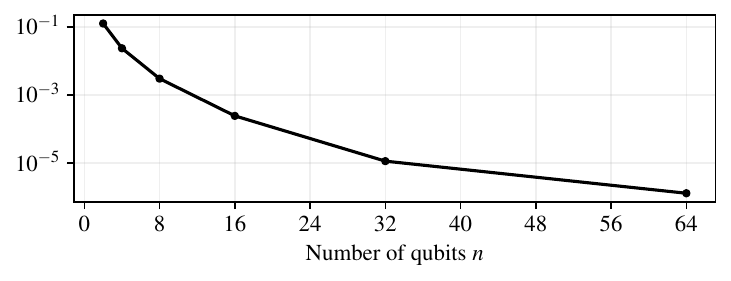}
		\caption{Tree tensor network ansatz \cite{grant2018hierarchical,zhao2021analyzingbarren} with $H = I^{\otimes n-1}\otimes Z$.}
		\label{fig:tree-variance}
	\end{subfigure}
	\caption{Computed variance for the two ans\"atze in Figure \ref{fig:ansatze} for increasing number of qubits. The gradient is w.r.t. the first parameter in each ansatz. Note that the y-axis is log scaled. Generating each of those curves took less than 60 seconds on a regular laptop.}
	\label{fig:variance}
\end{figure}

To illustrate our approach, we look at two example ans\"atze.
First, consider the ansatz in Figure \ref{fig:sim10} studied by Sim et al. \cite{sim2019expressibility}.
Figure \ref{fig:sim10-variance} plots the gradient variance of this ansatz for increasing $n$ when measuring in the computational basis, i.e. $H = Z^{\otimes n}$.
It appears like the variance is vanishing exponentially which suggests the existence of a barren plateau.
On the other hand consider the (discriminative) tree tensor network ansatz \cite{grant2018hierarchical,zhao2021analyzingbarren} shown in Figure \ref{fig:tree}.
In \cite{zhang2020toward} and \cite{zhao2021analyzingbarren} it was proven that tree tensor network ans\"atze do not exhibit barren plateaus and indeed the variance computed in Figure \ref{fig:tree-variance} does not seem to vanish exponentially.

We want to stress that interpreting graphs like the ones shown in Figure \ref{fig:variance} is of course not a formal proof for the existence or absence of barren plateaus.
Looking at a finite set of data points is not enough to conclusively judge the asymptotic behaviour of a function.
For example, it could be the case that the curve in Figure \ref{fig:sim10-variance} starts to flatten after some point $n_0 > 32$.
However, practically speaking, these numerical results show us the behaviour for circuit dimensions that are used in practice.
Focusing on realistic qubit numbers should thus give us a good idea about ansatz trainability in the real world.

\subsection{Benchmarks}

\begin{figure}
	\centering
	\includegraphics[scale=0.8]{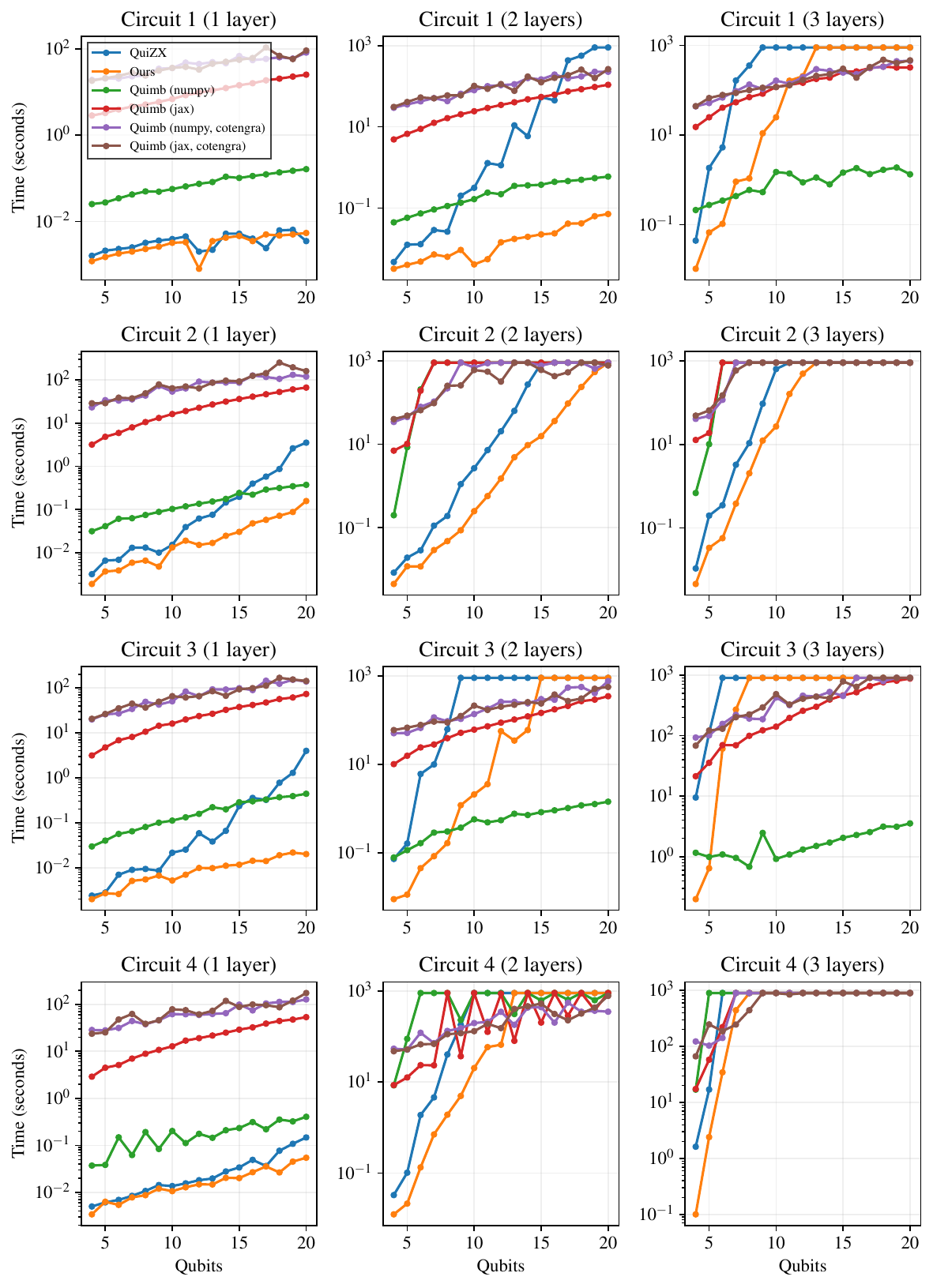}
	\caption{Running time of our method compared to QuiZX and quimb for computing gradient variances for increasing qubit counts.}
	\label{fig:barren-bench}
\end{figure}

Many authors employ numerical experiments like the ones in Section \ref{sec:examples} to verify their analytical results regarding the existence of barren plateaus (for example \cite{mcclean2018barren, grant2019initialization}).
However, typically gradient variances are estimated statistically by sampling the circuit at many random parameter points $\vec{\theta}$.
Note that this approach only yields noisy estimates of the variance, whereas our results are exact.
More importantly, the statistical method requires significantly more computational resources since the circuits need to be contracted very often.
Our method on the other hand only requires contracting a single diagram and thus offers a speed-up over the naive statistical approach.

We benchmark our method against contracting the variance diagram (\ref{variance}) using QuiZX\footnote{By representing the triangles using T-spiders according to (\ref{tri-T}).} and using the tensor network library \texttt{quimb}\footnote{By turning the ZX diagram into a tensor network.} \cite{gray2018quimb}.
The circuits used and further details on the benchmarking setup can be found in appendix \ref{sec:barren-bench}.
We observe that our method always outperforms QuiZX, whereas the comparison against \texttt{quimb} depends on the circuit.
This is in line with general comparisons between stabiliser decomposition and tensor contraction methods whose exponential scaling depend on different circuit characteristics.

\section{Conclusion and Future Work}

We built on the work in \cite{kissinger2022simulating,kissinger2022classical} by extending the QuiZX library with stabiliser decompositions and simplifications for triangles, which we represent using star edges in graph-like diagrams.
We demonstrate a simulation speed-up by multiple orders of magnitude for random Clifford+T+CCZ circuits and for a modified version of the hidden-shift circuits introduced by \cite{bravyi2016improved}.
Furthermore, we show that our implementation can be used for other tasks besides classical simulation.
Concretely, we numerically detect barren plateaus in parametrised quantum circuits by contracting algebraic ZX diagrams representing a circuit's gradient variance, which is significantly faster than the traditional approach of sampling from the parameter space and computing the variance statistically.

In the future, it would be interesting to see whether there are entangled states involving triangles that admit more efficient stabiliser decompositions (analogous to the cat states used in \cite{kissinger2022classical}), or whether it is possible to find partial stabiliser decompositions for triangles similar to \cite{kissinger2022classical}.
Furthermore, it would be interesting to investigate other potential simplification strategies for ZX diagrams with triangles/stars.
Finally, it might be worth investigating whether the diagrammatic stabiliser decomposition approach can be applied to approximate simulation. % (for example the stabiliser extent method \cite{bravyi2019simulation}).
This would be particularly useful for the use case of barren plateau detection, since there we are only interested in gauging whether the variance decays exponentially or not.
Hence, larger errors than for circuit simulation could be permissible here.

\subsection*{Acknowledgements}

We would like to thank Pablo Andrés-Martínez, Tuomas Laakkonen, and Michael Lubasch for their feedback on an earlier version of this manuscript.

%\nocite{*}
\bibliographystyle{eptcs}
\bibliography{bibliography}
%\bibliography{bibliography}
%\printbibliography

\appendix

\section{Benchmarking Details} \label{sec:barren-bench}

\begin{figure}
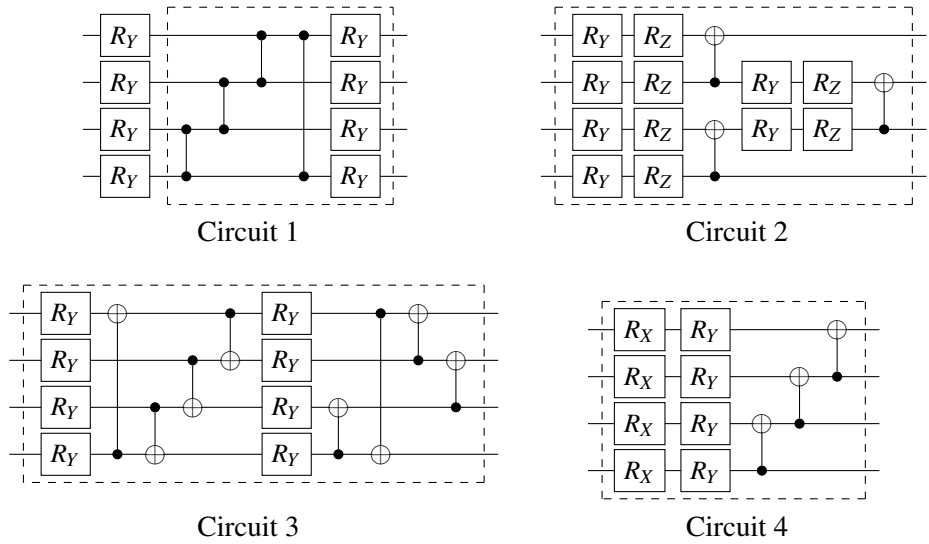

	\centering
	\begin{subfigure}[b]{0.4\textwidth}
		\centering
		\input{\tikzitpathfigs/barren/sim10-small.tikz}
		\\[5pt]
		{Circuit 1}
	\end{subfigure}
	\begin{subfigure}[b]{0.4\textwidth}
		\centering
		\input{\tikzitpathfigs/barren/sim11.tikz}
		\\[5pt]
		{Circuit 2}
	\end{subfigure}
	\\[15pt]
	\begin{subfigure}[b]{0.4\textwidth}
		\centering
		\input{\tikzitpathfigs/barren/sim15.tikz}
		\\[5pt]
		{Circuit 3}
	\end{subfigure}
	\begin{subfigure}[b]{0.4\textwidth}
		\centering
		\input{\tikzitpathfigs/barren/sim2.tikz}
		\\[5pt]
		{Circuit 4}
	\end{subfigure}
	\caption{Selection of ans\"atze from Sim et al. \cite{sim2019expressibility} used for benchmarking. The dashed box indicates the part of the circuit that can be repeated for multiple layers. Note that circuit 1 is identical to the one in Figure \ref{fig:sim10}.}
	\label{fig:sim-all}
\end{figure}

We benchmark our method for barren plateau detection against contracting the diagram (\ref{variance}) using QuiZX and the tensor network library \texttt{quimb} \cite{gray2018quimb}.
We run benchmarks for one, two, and three layers of the circuits shown in Figure \ref{fig:sim-all} which were studied in \cite{sim2019expressibility}.
For QuiZX, we replace the triangles in (\ref{variance}) with T-spiders according to (\ref{tri-T}).
For \texttt{quimb}, we turn the diagram into a tensor network using the built-in \texttt{to\_tensor} method from \texttt{pyzx}.
We compare the performance between the default \texttt{auto-hq} preset from \texttt{opt\_einsum} and \texttt{ReusableHyperOptimizer} from \texttt{cotengra} (\texttt{max\_repeats=16, reconf\_opts=\{\}}) as the contraction optimizer, as well as the performance between the \texttt{numpy} and \texttt{jax} numerical backends.

All simulations were run on a single CPU core and a 15 minute timeout was applied. 
When \texttt{quimb} fails to create an intermediate tensor of over 32 indices,
we consider that to be a out-of-memory timeout ($2^{32} * 128 \text{bits} > 65 \text{GB}$).
The results are shown in Figure \ref{fig:barren-bench}.

\section{Proofs} \label{sec:proofs}

We prove all equations with the axioms from the algebraic ZX-calculus \cite{wang2019algebraic} and rules applied correspond to equations in \cite{wang2019algebraic}.

\begin{lemma} \label{lem:star-had} \cite{wang2019algebraic}
	\[ \begin{tikzpicture}
	\begin{pgfonlayer}{nodelayer}
		\node [style=starv] (1) at (0, 0) {};
		\node [style=none] (2) at (0, 1) {};
		\node [style=none] (3) at (0, -1) {};
		\node [style=none] (4) at (1, 0) {$=$};
		\node [style=none] (5) at (3.25, 1.5) {};
		\node [style=none] (6) at (3.25, -1.5) {};
		\node [style=hadamard] (7) at (3.25, 1) {};
		\node [style=dtriangle] (8) at (3.25, 0.25) {};
		\node [style=Z box] (9) at (3.25, -0.75) {$\frac{1}{2}$};
		\node [style=none] (10) at (2, 0) {$\sqrt 2$};
	\end{pgfonlayer}
	\begin{pgfonlayer}{edgelayer}
		\draw (3.center) to (2.center);
		\draw (6.center) to (5.center);
	\end{pgfonlayer}
\end{tikzpicture}
 \]
\end{lemma}
\begin{proof}
	\[ \begin{tikzpicture}
	\begin{pgfonlayer}{nodelayer}
		\node [style=starv] (1) at (0, 0) {};
		\node [style=none] (2) at (0, 1) {};
		\node [style=none] (3) at (0, -1) {};
		\node [style=none] (4) at (1, 0) {$=$};
		\node [style=dtriangle, rotate=180] (5) at (2, -0.5) {};
		\node [style=X phase dot] (6) at (2, 0.5) {$\pi$};
		\node [style=none] (7) at (2, 1.25) {};
		\node [style=none] (8) at (2, -1.25) {};
		\node [style=none] (9) at (3, 0) {$=$};
		\node [style=dtriangle, rotate=180] (10) at (4, -1) {};
		\node [style=Z phase dot] (11) at (4, 0.5) {$\pi$};
		\node [style=none] (12) at (4, 1.75) {};
		\node [style=none] (13) at (4, -1.75) {};
		\node [style=hadamard] (14) at (4, 1.25) {};
		\node [style=hadamard] (15) at (4, -0.25) {};
		\node [style=none] (16) at (5.25, 0) {$=$};
		\node [style=dtriangle, rotate=180] (17) at (8, -2) {};
		\node [style=Z phase dot] (18) at (8, 2) {$\pi$};
		\node [style=none] (19) at (8, 3.25) {};
		\node [style=none] (20) at (8, -2.75) {};
		\node [style=hadamard] (21) at (8, 2.75) {};
		\node [style=dtriangle] (22) at (8, 1) {};
		\node [style=dtriangle, rotate=180] (23) at (8, -1) {};
		\node [style=Z box] (24) at (8, 0) {$-\frac{1}{2}$};
		\node [style=none] (25) at (8.5, -0.625) {\scriptsize $-1$};
		\node [style=none] (26) at (8.5, 1.25) {\scriptsize $-1$};
		\node [style=none] (27) at (6.5, 0) {$\sqrt 2$};
		\node [style=none] (28) at (9.5, 0) {$=$};
		\node [style=none] (31) at (11.75, 1.5) {};
		\node [style=none] (32) at (11.75, -1.5) {};
		\node [style=hadamard] (33) at (11.75, 1) {};
		\node [style=dtriangle] (34) at (11.75, 0.25) {};
		\node [style=Z box] (36) at (11.75, -0.75) {$\frac{1}{2}$};
		\node [style=none] (39) at (10.5, 0) {$\sqrt 2$};
		\node [style=none] (41) at (3, 0.5) {\scriptsize $H$};
		\node [style=none] (42) at (5.25, 0.5) {\scriptsize $EU$};
		\node [style=none] (43) at (9.5, 0.5) {\scriptsize $Inv$};
	\end{pgfonlayer}
	\begin{pgfonlayer}{edgelayer}
		\draw (3.center) to (2.center);
		\draw (8.center) to (7.center);
		\draw (13.center) to (12.center);
		\draw (20.center) to (19.center);
		\draw (32.center) to (31.center);
	\end{pgfonlayer}
\end{tikzpicture}
 \]
\end{proof}

\bigskip
\noindent
\textit{Proof of Lemma \ref{lem:hopf}.}
\begin{gather*}
	 \\
	 \\
	
\end{gather*}
\hfill$\square$

\bigskip
\noindent
\textit{Proof of Lemma \ref{lem:decomp-split}.}
\begin{align*} 
	 
	~&=~~ \begin{tikzpicture}
	\begin{pgfonlayer}{nodelayer}
		\node [style=Z phase dot] (0) at (1.75, 1) {$\gamma_1$};
		\node [style=none] (1) at (2.5, 1.5) {};
		\node [style=none] (2) at (2.5, 0.5) {};
		\node [style=none, rotate=90] (3) at (2.375, 1) {$...$};
		\node [style=Z phase dot] (4) at (1.75, -1) {$\gamma_m$};
		\node [style=none] (5) at (2.5, -0.5) {};
		\node [style=none] (6) at (2.5, -1.5) {};
		\node [style=none, rotate=90] (7) at (2.375, -1) {$...$};
		\node [style=none, rotate=90] (8) at (1.75, 0) {$...$};
		\node [style=Z phase dot] (9) at (-1.75, 1) {$\beta_1$};
		\node [style=none] (10) at (-2.5, 1.5) {};
		\node [style=none] (11) at (-2.5, 0.5) {};
		\node [style=none, rotate=90] (12) at (-2.375, 1) {$...$};
		\node [style=Z phase dot] (13) at (-1.75, -1) {$\beta_n$};
		\node [style=none] (14) at (-2.5, -0.5) {};
		\node [style=none] (15) at (-2.5, -1.5) {};
		\node [style=none, rotate=90] (16) at (-2.375, -1) {$...$};
		\node [style=none, rotate=90] (17) at (-1.75, 0) {$...$};
		\node [style=X dot] (18) at (-0.5, 1) {};
		\node [style=X dot] (19) at (-0.5, -1) {};
		\node [style=X dot] (20) at (0.5, 1) {};
		\node [style=X dot] (21) at (0.5, -1) {};
	\end{pgfonlayer}
	\begin{pgfonlayer}{edgelayer}
		\draw (2.center) to (0);
		\draw (0) to (1.center);
		\draw (6.center) to (4);
		\draw (4) to (5.center);
		\draw (11.center) to (9);
		\draw (9) to (10.center);
		\draw (15.center) to (13);
		\draw (13) to (14.center);
		\draw [style=star edge] (21) to (4);
		\draw [style=star edge] (0) to (20);
		\draw [style=hadamard edge] (9) to (18);
		\draw [style=hadamard edge] (19) to (13);
	\end{pgfonlayer}
\end{tikzpicture}
 ~~+~~ e^{i\alpha} ~ \begin{tikzpicture}
	\begin{pgfonlayer}{nodelayer}
		\node [style=Z phase dot] (0) at (2.25, 1) {$\gamma_1$};
		\node [style=none] (1) at (3, 1.5) {};
		\node [style=none] (2) at (3, 0.5) {};
		\node [style=none, rotate=90] (3) at (2.875, 1) {$...$};
		\node [style=Z phase dot] (4) at (2.25, -1) {$\gamma_m$};
		\node [style=none] (5) at (3, -0.5) {};
		\node [style=none] (6) at (3, -1.5) {};
		\node [style=none, rotate=90] (7) at (2.875, -1) {$...$};
		\node [style=none, rotate=90] (8) at (2.25, 0) {$...$};
		\node [style=Z phase dot] (9) at (-2, 1) {$\beta_1$};
		\node [style=none] (10) at (-2.75, 1.5) {};
		\node [style=none] (11) at (-2.75, 0.5) {};
		\node [style=none, rotate=90] (12) at (-2.625, 1) {$...$};
		\node [style=Z phase dot] (13) at (-2, -1) {$\beta_n$};
		\node [style=none] (14) at (-2.75, -0.5) {};
		\node [style=none] (15) at (-2.75, -1.5) {};
		\node [style=none, rotate=90] (16) at (-2.625, -1) {$...$};
		\node [style=none, rotate=90] (17) at (-2, 0) {$...$};
		\node [style=X phase dot] (18) at (-0.5, 1) {$\pi$};
		\node [style=X phase dot] (19) at (-0.5, -1) {$\pi$};
		\node [style=X phase dot] (20) at (0.75, 1) {$\pi$};
		\node [style=X phase dot] (21) at (0.75, -1) {$\pi$};
	\end{pgfonlayer}
	\begin{pgfonlayer}{edgelayer}
		\draw (2.center) to (0);
		\draw (0) to (1.center);
		\draw (6.center) to (4);
		\draw (4) to (5.center);
		\draw (11.center) to (9);
		\draw (9) to (10.center);
		\draw (15.center) to (13);
		\draw (13) to (14.center);
		\draw [style=star edge] (21) to (4);
		\draw [style=star edge] (0) to (20);
		\draw [style=hadamard edge] (9) to (18);
		\draw [style=hadamard edge] (19) to (13);
	\end{pgfonlayer}
\end{tikzpicture}
 \\[5pt]
	~&=~~ \frac{1}{\sqrt 2^{n}}~  ~~+~~ \frac{e^{i\alpha}}{\sqrt 2^{n+m}}~ 
\end{align*}
\hfill$\square$

\bigskip
\noindent
\textit{Proof of Lemma \ref{lem:simp/copy}.}
\begin{gather*}
	\input{\tikzitpathfigs/simp/copy-1.tikz} 
	~=~ \input{\tikzitpathfigs/simp/copy-2.tikz} 
	~=~ \frac{1}{\sqrt 2^{n-1}}~ \input{\tikzitpathfigs/simp/copy-3.tikz}
	\\
	\input{\tikzitpathfigs/simp/copy-pi-1.tikz} 
	~=~ \input{\tikzitpathfigs/simp/copy-pi-2.tikz} 
	~=~ \frac{e^{i\alpha}}{\sqrt 2^{n+m-1}}~ \input{\tikzitpathfigs/simp/copy-pi-3.tikz}
\end{gather*}

\begin{lemma}
	It is not possible to construct the state $\begin{tikzpicture}
	\begin{pgfonlayer}{nodelayer}
		\node [style=none] (0) at (0, -0.5) {};
		\node [style=Z box] (1) at (0, 0.25) {$\sqrt 2$};
		\node [style=none] (2) at (1.25, 0) {$=$};
		\node [style=Z phase dot] (3) at (2.5, 1.25) {$\frac{\pi}{2}$};
		\node [style=dtriangle] (4) at (2.5, 0.25) {};
		\node [style=Z phase dot] (5) at (2.5, -0.75) {$-\frac{\pi}{4}$};
		\node [style=none] (6) at (2.5, -1.5) {};
	\end{pgfonlayer}
	\begin{pgfonlayer}{edgelayer}
		\draw (0.center) to (1);
		\draw (6.center) to (3);
	\end{pgfonlayer}
\end{tikzpicture}
$ in the Clifford+Triangle fragment.
\end{lemma}
\begin{proof}
	Suppose there is a Clifford+Triangle diagram with $\begin{tikzpicture}
	\begin{pgfonlayer}{nodelayer}
		\node [style=none] (0) at (0, -0.5) {};
		\node [style=none] (1) at (-0.75, 0) {};
		\node [style=none] (2) at (0.75, 0) {};
		\node [style=none] (3) at (0, 0.75) {};
		\node [style=none] (4) at (0, 0) {};
		\node [style=none] (5) at (0, 0.3) {\scriptsize $R$};
	\end{pgfonlayer}
	\begin{pgfonlayer}{edgelayer}
		\draw (1.center) to (2.center);
		\draw (2.center) to (3.center);
		\draw (3.center) to (1.center);
		\draw (4.center) to (0.center);
	\end{pgfonlayer}
\end{tikzpicture}
 = $.
	Let $h$ be the number of Hadamards that occur in $R$.
	Using the equality
	\[ \begin{tikzpicture}
	\begin{pgfonlayer}{nodelayer}
		\node [style=none] (0) at (-2.5, 0) {};
		\node [style=none] (1) at (-1, 0) {};
		\node [style=hadamard] (2) at (-1.75, 0) {};
		\node [style=none] (3) at (0, 0) {$=$};
		\node [style=none] (4) at (2, 0) {};
		\node [style=none] (5) at (5.5, 0) {};
		\node [style=X dot] (6) at (3.75, 0) {};
		\node [style=Z phase dot] (7) at (2.75, 0) {$\frac{\pi}{2}$};
		\node [style=Z phase dot] (8) at (4.75, 0) {$\frac{\pi}{2}$};
		\node [style=Z phase dot] (9) at (3.75, 1) {$-\frac{\pi}{2}$};
		\node [style=none] (10) at (1.25, 0) {$\frac{1}{\sqrt 2}$};
	\end{pgfonlayer}
	\begin{pgfonlayer}{edgelayer}
		\draw (0.center) to (1.center);
		\draw (4.center) to (6);
		\draw (6) to (5.center);
		\draw (9) to (6);
	\end{pgfonlayer}
\end{tikzpicture}
 \]
	we can turn $R$ into a diagram $R'$ that only contains triangles, pink spiders, and green spiders such that $ = \sqrt 2^{-h}\begin{tikzpicture}
	\begin{pgfonlayer}{nodelayer}
		\node [style=none] (0) at (0, -0.5) {};
		\node [style=none] (1) at (-0.75, 0) {};
		\node [style=none] (2) at (0.75, 0) {};
		\node [style=none] (3) at (0, 0.75) {};
		\node [style=none] (4) at (0, 0) {};
		\node [style=none] (5) at (-0.05, 0.3) {\scriptsize $R'$};
	\end{pgfonlayer}
	\begin{pgfonlayer}{edgelayer}
		\draw (1.center) to (2.center);
		\draw (2.center) to (3.center);
		\draw (3.center) to (1.center);
		\draw (4.center) to (0.center);
	\end{pgfonlayer}
\end{tikzpicture}
$.
	Obviously, the matrix for $R'$ can only contain elements of the ring $\mathbb{Z}[i]$. Consequently we have both $R'\ket{0} \in \mathbb{Z}[i]$ and $R'\ket{1} \in \mathbb{Z}[i]$.
	Thus, we must also have 
	\[ 
	\begin{tikzpicture}
	\begin{pgfonlayer}{nodelayer}
		\node [style=X dot] (0) at (0, -0.75) {};
		\node [style=none] (1) at (-0.75, 0) {};
		\node [style=none] (2) at (0.75, 0) {};
		\node [style=none] (3) at (0, 0.75) {};
		\node [style=none] (4) at (0, 0) {};
		\node [style=none] (5) at (0, 0.3) {\scriptsize $R$};
	\end{pgfonlayer}
	\begin{pgfonlayer}{edgelayer}
		\draw (1.center) to (2.center);
		\draw (2.center) to (3.center);
		\draw (3.center) to (1.center);
		\draw (4.center) to (0);
	\end{pgfonlayer}
\end{tikzpicture}
, \begin{tikzpicture}
	\begin{pgfonlayer}{nodelayer}
		\node [style=X phase dot] (0) at (0, -0.75) {$\pi$};
		\node [style=none] (1) at (-0.75, 0) {};
		\node [style=none] (2) at (0.75, 0) {};
		\node [style=none] (3) at (0, 0.75) {};
		\node [style=none] (4) at (0, 0) {};
		\node [style=none] (5) at (0, 0.3) {\scriptsize $R$};
	\end{pgfonlayer}
	\begin{pgfonlayer}{edgelayer}
		\draw (1.center) to (2.center);
		\draw (2.center) to (3.center);
		\draw (3.center) to (1.center);
		\draw (4.center) to (0);
	\end{pgfonlayer}
\end{tikzpicture}
 
	\begin{cases}
			\in \mathbb{Q}[i] & \text{ if $h$ is even} \\
			\not\in \mathbb{Q}[i] & \text{ if $h$ is odd}
		\end{cases}
	\]
	However, this cannot be true since
	\[ \begin{tikzpicture}
	\begin{pgfonlayer}{nodelayer}
		\node [style=X dot] (0) at (0, -0.75) {};
		\node [style=Z box] (1) at (0, 0.25) {$\sqrt 2$};
	\end{pgfonlayer}
	\begin{pgfonlayer}{edgelayer}
		\draw (0) to (1);
	\end{pgfonlayer}
\end{tikzpicture}
 = 1 \in \mathbb{Q}[i] \qquad\qquad \begin{tikzpicture}
	\begin{pgfonlayer}{nodelayer}
		\node [style=X phase dot] (0) at (0, -0.75) {$\pi$};
		\node [style=Z box] (1) at (0, 0.25) {$\sqrt 2$};
	\end{pgfonlayer}
	\begin{pgfonlayer}{edgelayer}
		\draw (0) to (1);
	\end{pgfonlayer}
\end{tikzpicture}
 = \sqrt 2 \not\in \mathbb{Q}[i] \]
\end{proof}

\begin{corollary} \label{cor:no-magic}
	It is not possible to construct magic states in the Clifford+Triangle fragment.
\end{corollary}
\begin{proof}
	Suppose we could construct magic states. Then we could also construct
	\[  \]
\end{proof}

\end{document}